\documentclass[11pt]{article}
\usepackage[T1]{fontenc}
\usepackage{amsfonts}
\usepackage{amsmath}
\usepackage{amssymb}
\usepackage{amsthm}
\usepackage{bbm}
\usepackage{bm}
\usepackage{mathrsfs}
\usepackage{verbatim}
\usepackage{setspace}
\usepackage{color}
\usepackage{pdfsync}
\usepackage{enumitem} 
\usepackage{graphicx}
\usepackage[normalem]{ulem}
\usepackage[comma,authoryear]{natbib}
\usepackage{setspace}

\theoremstyle{plain}
\newtheorem{theorem}{Theorem}[section]
\newtheorem{proposition}[theorem]{Proposition}
\newtheorem{lemma}[theorem]{Lemma}

\theoremstyle{definition}

\newtheorem{remark}[theorem]{Remark}

\newtheorem{example}[theorem]{Example}

\theoremstyle{remark}
\newcommand{\R}{\mathbb{R}}

\newcommand{\cA}{\mathcal{A}}

\newcommand{\cI}{\mathcal{I}}

\newcommand{\cL}{\mathcal{L}}

\DeclareMathOperator{\sgn}{sign}
\DeclareMathOperator{\tr}{Tr}

\newcommand{\as}{\mbox{-a.s.}}

\newcommand{\sta}{\textnormal{sta}}
\newcommand{\dyn}{\textnormal{dyn}}

\numberwithin{equation}{section}

\usepackage[pdfborder={0 0 0}]{hyperref}
\hypersetup{
  urlcolor = black,
  pdfauthor = {Marcel Nutz, Jose A. Scheinkman},
  pdfkeywords = {Speculation; Supply; Short-selling; Resale Option; Delay option},
  pdftitle = {Shorting in Speculative Markets},
  pdfsubject = {Shorting in Speculative Markets},
  pdfpagemode = UseNone
}
\begin{document}

\title{\vspace{-3.5em}
%Supply and Shorting in Speculative Markets%
Shorting in Speculative Markets%
\footnote{The authors have no conflict of interest relevant to this paper. We are indebted to Pete Kyle and Xunyu Zhou for fruitful discussions and to two referees, an Associate Editor,  and Stefan Nagle, the Editor, for comments and suggestions that greatly improved the paper.  An earlier version was titled ``Supply and Shorting in Speculative Markets.''}
}
\date{\today}
\author{
  Marcel Nutz%
  \thanks{
  Columbia University. Research supported by an Alfred P.\ Sloan Fellowship and NSF Grants DMS-1512900 and DMS-1812661.}
  \and
  Jos\'e A.\ Scheinkman%
  \thanks{
  Columbia University, Princeton University and NBER. }
}
\maketitle \vspace{-1.3em}

\onehalfspacing
\begin{abstract}
We propose a continuous-time model of trading with heterogeneous beliefs. Risk-neutral agents face  quadratic costs-of-carry on positions and thus their marginal valuations decrease with the size of their position, as it would be the case for risk-averse agents. In the equilibrium models of  heterogeneous beliefs that followed Harrison--Kreps, investors are risk-neutral, short-selling is prohibited  and agents face  constant marginal costs of carrying positions. The resulting resale option guarantees that  the  price exceeds the price of the asset when speculation is ruled out; the difference is identified as a \emph{bubble}. In our model increasing marginal costs entail that the price depends on  asset supply. Second, agents  also value an option to delay, and this may cause the market to equilibrate \emph{below} the  buy-and-hold price. Third, we introduce the possibility of short-selling. A Hamilton--Jacobi--Bellman equation of a novel form quantifies precisely the influence of the costs-of-carry on the price. An unexpected decrease in shorting costs may lead to the collapse of a bubble; this links the financial innovations  that facilitated shorting of MBSs to the subsequent collapse of prices. \end{abstract}

%%%%%%%%%%%%%%%%%%%%%%

\newpage
\onehalfspacing
\section{Introduction}

As \cite{KA2005} elaborate, many classical economists argued that the purchase of securities for re-sale rather than for investment income was the engine that generated asset price bubbles. To explain such \emph{speculation} in a dynamic equilibrium model, \cite{harrisonkreps78} study risk-neutral agents with fluctuating heterogeneous beliefs. In their model, long positions can be financed at a constant interest rate and short-selling is ruled out. The buyer of an asset thus acquires both a stream of future dividends and an option to resell which in combination with fluctuating beliefs guarantees that speculators are willing to pay more for an asset than they would pay if they were forced to hold the asset to maturity; that is, risk-neutral investors would pay to be allowed to speculate. \cite{scheinkmanxiong03} considered a model where heterogeneous beliefs result from agents' overconfidence on different public signals and added trading costs, and showed that these models  generate a correlation between trading volume and overpricing,\footnote{See also \cite{BerestyckiBruggemaMonneauScheinkman.16}.} a characteristic of several major bubble episodes in the last three centuries.\footnote{See e.g.\ \cite{CNW06} on the South Sea bubble,  \cite{HS07} on the Roaring Twenties, \cite{ofek03} and \cite{Cochrane02} on the internet bubble, \cite{xiongyu} on the Chinese warrant bubble.}

Another stylized fact is that bubble implosions often follow increases in supply. For instance, the implosion of the dotcom bubble was preceded by a vast increase in the float of internet shares.\footnote{See \cite{ofek03}.}
The South Sea bubble lasted less than one year, but in that  period the amount of outstanding shares of the South Sea Company more than doubled, and many other joint-stock companies were established.\footnote{The directors of the SSC  understood  that  bubble companies competed with the SSC's conversion scheme and could deflate its own bubble. \cite{Harris} documents that  the Bubble Act of 1720, which banned joint-stock companies except if authorized by Royal Charter, was produced at the behest of the company to limit the competition for capital.}  However, the assumption of risk-neutral investors facing constant marginal costs in the earlier literature on disagreement and bubbles implies that supply is \emph{irrelevant} for pricing.\footnote{Except for the assumption of positive net supply, issues concerning supply of the asset subject to bubbles are also ignored in the rational bubbles literature (\cite{SW97}).} \cite{hongetal06} analyzed a two-period model with risk-averse investors where unexpected increases in supply could deflate bubbles. The economics is straightforward---when agents are risk-averse, their marginal valuation for a risky asset decreases with the amount they hold.

Short-selling an asset can be seen as a source of additional supply. The collapse of  prices for  mortgage backed securities in 2007 was preceded  by a series of financial innovations  that facilitated shorting: the creation of a standardized CDS for MBS in 2005, the introduction of traded indexes for subprime mortgage backed credit derivatives in 2006, and the use of  CDS to construct synthetic CDOs that allowed Wall Street to satisfy the global demand for US AAA mortgage bonds without going through the relatively slow process of originating new mortgages.\footnote{See \cite{scheinkman14} for a summary or \cite{lewis15} for an excellent detailed account.} 
%\footnote{\cite{FGW} lay out 12 facts about the mortgage market during the boom years and argue persuasively that  ``any
%story consistent with the 12 facts must have overly optimistic beliefs about house prices
%at its center.''} 
The amounts shorted were substantial; \cite{CHW11} estimate that synthetic CDOs, mostly issued after 2005H2,  more than doubled the amount of  BBB Home Equity Bonds placed in CDOs during 1998--2007.\footnote{BBB tranches of Home Equity Bonds  were an essential fuel to the CDO machine that transformed subprime mortgages into AAA rated bonds.} It is unlikely that this supply would have been absorbed without any price impact. In any case,  starting in the second half of 2007, prices seem to exhibit substantial  discounts relative to fundamentals.\footnote{\cite{BCT17}  provide a methodology to calculate the \emph{intrinsic} value of a CDO and apply it to market data (see their Appendix A).  They  attribute the low prices to the increase in asymmetric information between buyers and sellers that followed  the   downgrades of MBS securities by rating agencies in summer 2007. Analyzing the pricing of index credit default swaps post-crisis, \cite{SW11} suggest that the pricing reflected a limited  supply of insurance of asset backed securities, presumably relative to demand.}

In this paper we propose a finite-horizon continuous-time  model with $n$ types of investors trading a single asset  and aiming to maximize expected cumulative net gains from trade. These investors are risk-neutral, face a constant interest rate, and have fluctuating heterogeneous beliefs about the evolution of a Markov state that  determines the asset's payoff.  In contrast to the previous literature, shorting is allowed.  Investors pay costs which are proportional to the square of their positions but
%
%
%\footnote{{The assumption that costs are proportional to the square of the position does not accommodate the fact that borrowers of stock typically pay a fee quoted as an annualized percentage of the value of the loaned securities (the rebate rate). This assumption is made to simplify the exposition and allow us to concentrate on the effects of the size of a position on an agent's marginal valuation. In Appendix~\ref{se:linearCosts} we discuss how our results can be generalized if an additional linear term is added to the costs-of-carry. \sout{The agents are then divided into three rather than two groups: while strongly optimistic and pessimistic agents hold long and short positions, respectively, an intermediate group expects that prices will not move enough to compensate for the first-order costs and stays out. However,} The qualitative properties of equilibria remain the same as without the linear term, so we have chosen to use the simpler model in the body of the paper.}} of their positions.\footnote{For tractability, we assume that costs are a function of the size rather than the value of the position. See Appendix~\ref{se:costsOnValue} for a discussion of costs as a function of position values.}
%
the  constant of proportionality that defines the cost of going short may be larger than the corresponding constant for going long. This asymmetry between costs of going short and long is a well-known feature of financial markets, see e.g. \cite{davolio},
%\footnote{Some of this asymmetry refers to costs which are linear on positions; we treat these differences in Appendix~\ref{se:linearCosts}.} 
and the assumptions in the earlier literature correspond to an infinite cost for short positions and constant marginal costs for long positions.
On the one hand, the costs in our model reflect monetary costs of holding a position (in particular increasing costs of capital). On the other hand, they stand in for risks that we do not explicitly model, such as market-wide liquidity shocks that would force agents to liquidate their positions at unfavorable prices or the recall-risk faced by short-sellers.

Indeed, since costs are quadratic, an agent's marginal valuation of the asset will decrease as their position increases, as it would be the case for risk-averse agents. Thus we view our setup as an alternative to a much less tractable model with risk aversion, with many of the same forces present.\footnote{One difference from models of disagreement that use risk aversion to avoid no-shorting constraints is that the presence of holding costs allows for equilibrium to exist even when agents disagree about  perceived arbitrage opportunities.}
In particular, we show below that an increase in aggregate supply of the asset lowers equilibrium prices. 
Importantly, using the two cost coefficients as separate parameters allows us to impose asymmetric costs and study the impact of changes in relative costs on prices. 
%If the costs on long positions are interpreted as approximating risk aversion, the larger costs on shorts approximate the combined effect of risk aversion and short-specific costs. 
By contrast, traditional models with risk aversion either treat longs and shorts symmetrically or completely rule out shorts by imposing portfolio constraints. 

We examine the equilibrium price for the asset, given by a function of time and current state. Types that expect prices to increase on average over the next instant choose to go long; the size of their position depends on the difference between their expected price changes and the marginal cost of carrying long positions.  The other types choose to go short, by amounts that depend on their expected price changes and the cost of carrying short positions. Equilibrium requires that the longs absorb the shorts plus an exogenous supply. Theorem~\ref{th:PDE} below shows that there is a unique equilibrium price function and that it can be characterized by a partial differential equation (PDE). 
This equation is of Hamilton--Jacobi--Bellman-type with a novel form where the optimization runs over the ways to divide the agents into two groups at any time and state; at the optimum, these are the optimists (holding long positions in equilibrium) and the pessimists (holding shorts).
A noteworthy feature is that the supply enters mathematically as a \emph{running cost} (i.e., like intermediate consumption in Merton's problem). Theorem~\ref{th:PDE} also quantifies precisely how these costs amplify or diminish the impact of optimists' and pessimists' views; for instance, as shorting gets more expensive relative to being long, optimists have a larger impact on the price.

We show that an increase in supply decreases the equilibrium price and that a decrease (increase) in the cost for long (short) positions increases the price. We also establish that as the cost for long positions converges to zero, the equilibrium price function converges to a function that does not depend on the cost of holding a short position or the amount supplied.  On the other hand, as the cost of shorting becomes prohibitive, the equilibrium price converges to a function that does depend on the cost of carrying long positions as well as the exogenous supply of the asset. 

To discuss the impact of speculation we first characterize the static equilibrium price; that is, the price that prevails when re-trading is not allowed and agents are forced to use buy-and-hold strategies. The difference between the dynamic price (where re-trading is possible) and the static price has been identified as the size of a bubble in the previous literature. 
A  buyer of the asset today may forecast that at some future date she would be able to sell at a price that would exceed her own valuation of the asset at that date. Because of this \emph{resale option}, she may be willing to pay more than what she believes is the discounted value of the payoff of an asset. In the classical models, this option causes equilibrium prices to exceed the price that would prevail if re-trading is ruled out. In addition, there is also an \emph{option to delay} which has not been highlighted in the earlier literature on heterogeneous beliefs. A speculator that buys $y$ units today may plan to buy additional units of the asset in future states of the world where there would be a larger difference between the asset price and her marginal valuation for the asset if she holds $y$ units. However, if the marginal cost of holding a long position is constant, this delay option has no impact in equilibrium. The intuition is that, since agents are risk-neutral and the marginal  cost-of-carry is constant, a buyer of a positive amount of the asset must be indifferent as to the amount of the asset she buys. Hence, the delay option has no value for this buyer, and in particular, the dynamic equilibrium price cannot be smaller than the static price. We  prove that this comparison holds even if shorting is allowed (see Proposition~\ref{pr:limitLongComparison}). We show through an example that in the presence of increasing marginal costs of going long, the option to delay may outweigh the resale option and cause the dynamic price to be \emph{lower} than the static one, even when shorting is prohibited. Thus the crucial assumption in the earlier literature that delivers the result that speculative prices exceed static prices is not the prohibition of short-sales, but the assumption of constant marginal costs for carrying long positions. 

When shorting is allowed, the short party also has  corresponding options. An agent that acquires a short position today may forecast that at some future date he would be able to repurchase the asset at a price that would be below his own valuation at that date.  This option to resell a short position, that is the option to cover shorts, in turn,  decreases the minimum amount pessimists would be willing to receive for shorting the asset, putting downward pressure on prices. The short party also enjoys an option to delay. In Example~\ref{ex:symmetricCosts} we show that when the cost of holding a short position is close to the cost of holding a long position,  the equilibrium price may be less than the static price. We argue that this happens because the long party values the resale option less than the short party values the repurchase option.  Example~\ref{ex:symmetricCosts} can also be used to illustrate that an unexpected decrease of the cost of shorting can lead to a collapse of an asset price bubble, thus rationalizing a link between the decrease of the cost of shorting in the MBS market in 2005--07 and the collapse of CDO prices.

The equilibrium in our model is not first-best except in the limit case of homogeneous beliefs. We show that
the equilibrium price and allocations obtain as solutions to the problem of a time-consistent planner who subsidizes and taxes the cost-of-carry to maximize the initial price.  We use this planner's problem to explain the structure of the PDE that characterizes equilibrium prices.

Our paper connects to a number of other contributions in the literature. In a pioneering paper allowing for short-sales and risk aversion in a continuous-time setting of heterogeneous beliefs,  \cite{dumasetal09}  consider a  ``complete markets'' model with two classes of agents, one of which is overconfident about a public signal. Overconfident-investors' reaction to the signal  introduces a risk factor---sentiment risk---which  carries a risk-premium and causes stock prices to be excessively volatile.  In their model, the delay option must be  valuable and supply affects equilibrium prices, but these questions are not explicitly analyzed. Instead, \cite{dumasetal09} focus on the important question of identifying the trading strategy that would allow a rational investor to take advantage of excessive stock price volatility and sentiment fluctuations and show that rational investors choose a conservative portfolio that is sensitive to their predictions about future realization of sentiment. The related paper by \cite{david} assumes the existence of  distinct processes for output and dividends of a zero-net-supply stock. The drifts of these processes  are given by an unobserved, finite state, continuous-time Markov chain. Agents agree to disagree on the probabilities of transitions across states and use the zero-supply stock to speculate against each other---creating an additional source of risk.   The focus of \cite{david} is the relationship between the equity premium and the time variation in agents' consumption. In these papers there is a cost symmetry between going short and long (the short party receives the equilibrium price and must pay the dividends that accumulate until the position is closed). Thus one cannot examine the effect of changes in shorting costs, which is the main motivation for our paper.

The literature on asset pricing with search frictions that follows \cite{DGP05} assumes that agents have fluctuating private benefits from holding an asset and that opportunities for trading are randomly distributed. Strict concavity of private benefits implies that the supply of the asset affects equilibrium prices.  Our assumption of quadratic costs of holding a position could be similarly motivated as private benefits (but in our case the differences come from heterogeneity of beliefs rather than benefits). The fluctuating private benefits also generate options to resell and to delay trading which are discussed in \cite{LR}, \cite{feldhutter} and \cite{HLW}. These authors  point out that depending on the curvature of  private benefits an increase in trading opportunities may decrease or increase the price of the asset. 
In particular the comparison between prices that would prevail  when re-trading opportunities are more or less frequent is ambiguous. However, short-selling or changes in shorting costs are not emphasized in this literature.

\cite{DGP02}\footnote{See also \cite{VW08}.} highlight the mechanics of shorting that prevails in  markets where shorts  pay a fee to borrow assets from longs.\footnote{Synthetic CDOs allowed pessimists to short  CDOs without borrowing the underlying securities.}  Agents disagree on the expected final payoff of an asset but there is no fluctuation of beliefs and hence no speculative behavior; all purchases are buy-and-hold. The dynamics arises because pessimists must meet longs and borrow their shares and these meetings occur with an intensity $\lambda$ per unit of time. In the model of \cite{DGP02}  an increase in supply decreases equilibrium prices (see Proposition 5). There are no explicit costs of shorting, but an increase in the intensity of meetings, $\lambda,$ has an ambiguous effect on prices; it lowers prices because the supply by shorts increases but it increases prices because longs are more likely to lend their shares. 

\cite{CM11} also study heterogeneous agents in an equilibrium model, with a focus on survival and market impact. They find that long-run price and portfolio impact are equivalent to the survival of an agent under different measures. Again, there is no fluctuation of beliefs; agents are optimists or pessimists because they over- or underestimate the (constant) drift parameter of dividends. By contrast, in our model, the notion of optimism is endogenous and state-dependent.

\cite{fostelgeanokoplos12} study a collateral equilibrium\footnote{\cite{GZ97}} in a model of financial innovation with heterogeneous beliefs. The introduction of a CDS  leads to a fall in the price of the underlying security, the fall being  more dramatic if tranching of the security is  already present. The introduction of this new derivative affects equilibrium prices but in the model of \cite{fostelgeanokoplos12} the initial beneficiaries of the new contract are the optimists, who in their language benefit from ``tranching cash.'' By contrast, in our model the initial beneficiaries of this cost decrease are the pessimists. \cite{lewis15} reports how starting in early 2005 a small group of traders who had pessimistic views on the housing market lobbied ISDA---the trade organization of over-the-counter market participants---to create  standardized CDS contracts on  mortgage-market securities that facilitated shorting.\footnote{See e.g.\ \cite{lewis15} pp.\ 48--50 on the creation of standardized CDS on MBS.}
\cite{OZ16} also examine the effect of introducing a CDS, in a model where traders differ on their horizons and beliefs. They postulate a per-unit cost for trading bonds which affects equally long and short positions, while CDS trading  is free. Thus the  introduction  of a CDS lowers the cost for both longs and shorts. It leads former bond buyers to switch to protection selling and former bond shorters to buying protection. Moreover long-horizon traders now hold a long position on the bond while buying protection.  The net effect may be an \emph{increase} of   bond prices. Although the mechanism described in \cite{OZ16} may have played a role in the CDO market, the sharp drop in prices of CDO-tranches suggests that it was overwhelmed by the lowering of the cost of shorting.

The paper is organized as follows. Section~\ref{se:main} contains the formulation of the problem and the characterization of the equilibrium as the solution to a Hamilton--Jacobi--Bellman  equation. Section~\ref{se:comparativeStatics} presents  comparative statics and limiting results. Section~\ref{se:speculation} discusses the role of speculation, while Section~\ref{se:planner} deals with the planner's problem.  Section~\ref{se:conclusion} concludes, and the Appendix details  proofs and several extensions of our model.

\section{Equilibrium Price}\label{se:main}

In this section, we detail our formal setup and show that it leads to a unique equilibrium. The equilibrium price is described by a partial differential equation of the Hamilton--Jacobi--Bellman type.

%%%%%%%%%%%%%%%%%%%%%%%%%%%
\subsection{Definition of the Equilibrium Price}

We consider $n\geq1$ types, each with a unit measure of agents,  who trade a security over a finite time interval~$[0,T]$. For brevity, we will often refer to a type as an agent. The security has a single payoff $f(X(T))$ at the horizon, where $f:\R^{d}\to\R$ is a bounded continuous function and $X(\omega)$, $\omega \in \Omega$ is the $d$-dimensional state process.\footnote{More precisely  we  take~$X$ to be the coordinate-mapping process on the space $\Omega=C([0,T],\R^{d})$ of continuous, $d$-dimensional paths, equipped with the canonical filtration and sigma-field. In what follows all processes are assumed to be progressively measurable.}  While there is no ambiguity about $f$, the agents agree to disagree on the evolution of the state process.  The views of agent $i$ are represented by a probability measure $Q_{i}$ on $\Omega$ under which~$X$ follows the stochastic differential equation (SDE)
\begin{equation}\label{eq:SDEforX}
  dX(t) = b_{i}(t,X(t))\,dt + \sigma_{i}(t,X(t))\,dW_{i}(t), \quad X(0)=x,
\end{equation}
where $W_{i}$ is a Brownian motion of dimension $d'$ and the functions
$$
  b_{i}: [0,T]\times \R^{d}\to \R^{d},\quad \sigma_{i}: [0,T]\times \R^{d}\to \R^{d\times d'} 
$$
are deterministic. We assume\footnote{These conditions could be relaxed considerably. The present form allows for a simple exposition avoiding issues of technical nature.
} throughout that (the components of) $b_{i}$ and $\sigma_{i}$ are in $C^{1,2}_{b}$, the 
set of all bounded continuous functions $g:[0,T]\times \R^{d}\to\R$ whose partial derivatives $\partial_{t}g$, $\partial_{x_{i}} g$, $\partial_{x_{i}x_{j}} g$ exist and are continuous and bounded on $[0,T)\times \R^{d}$. Moreover, we assume that\footnote{
Given a matrix $A$, we write $A^{2}$ for the product $AA^{\top}$ of $A$ with its transpose $A^{\top}$.
} $\sigma_{i}^{2}$ is uniformly parabolic; that is, its eigenvalues are uniformly bounded away from zero.
% or equivalently, there exists a constant $C>0$ such that $\xi^{\top}\sigma_{i}\sigma_{i}^{\top}(t,x)\xi\geq c|\xi|^{2}$ for all $\xi\in\R^{d}$ and all $(t,x)\in [0,T]\times \R^{d}$.
These conditions imply in particular that the SDE~\eqref{eq:SDEforX} has a unique (strong) solution.

Notice that we allow for the differences in beliefs to affect both drift and diffusion coefficients. As the volatility is more amenable to statistical estimation than the rate of return, the differences on drifts will typically be more significant. While much of the literature on disagreement in asset markets deals with constant volatility processes and thus naturally assumes perfect agreement on volatilities,
%that disagreement concerns only drifts, 
there is plenty of evidence that more complex processes involving stochastic and time-varying volatility are necessary to understand empirical features of asset prices. In this context, it is quite plausible, as argued by \cite{EJ03}, that agents may also differ in their forecasts\footnote{Note that while the past trajectory of $\sigma(t,X(t))$ can be inferred from the observation of $X(t)$, agents may very well differ in their forecasts. This is obvious if $\sigma$ depends on time $t$, but even if not, the past observation of $\sigma(t,X(t))$ will typically reveal little of the function $\sigma$ when $X$ is non-recurrent (e.g., of dimension larger than 2). We will see in Theorem~\ref{th:PDE} that the pricing of the security indeed depends on the future volatility over the entire time interval. This is quite natural as the same would be true in standard risk-neutral pricing when $f$ is a derivative on a stock, for instance.} for the volatility. In particular, it seems worthwhile to establish that our equilibrium is robust with respect to such differences. Nonetheless all our results have interest and all our examples are valid when agents disagree only about drifts.

Agents  trade the security  throughout the interval $[0,T]$, at a time $t$ price $P(t)$ to be determined in equilibrium.  The agents are subject to an instantaneous  cost-of-carry $c$ which is different for long and short positions,\footnote{The assumption that costs are proportional to the square of the position does not accommodate the fact that borrowers of stock may pay a fee quoted as an annualized percentage of the value of the loaned securities (the rebate rate). This assumption is made to simplify the exposition and allow us to concentrate on the effects of the size of a position on an agent's marginal valuation. In Appendix~\ref{se:linearCosts} we discuss how our results can be generalized if an additional linear term is added to the costs-of-carry. 
%The agents are then divided into three rather than two groups: while strongly optimistic and pessimistic agents hold long and short positions, respectively, an intermediate group expects that prices will not move enough to compensate for the first-order costs and stays out. However, The qualitative properties of equilibria remain the same as without the linear term, so we have chosen to use the simpler model in the body of the paper. 
In addition, for tractability, we assume that costs are a function of the size rather than the value of the position. See Appendix~\ref{se:costsOnValue} for a discussion of costs as a function of position values.}
\begin{equation}\label{eq:costOfCarry}
  c(y) = 
  \begin{cases}
    \frac{1}{2\alpha_{+}} y^{2},  & y\geq0,\\
    \frac{1}{2\alpha_{-}} y^{2},  & y<0.
  \end{cases}
\end{equation}
Here the (inverse) cost coefficients $\alpha_{\pm}$ are given constants\footnote{We examine in Appendix~\ref{se:heterogeneousCosts} how the model is altered if costs of going long and/or short vary across agents.
%While the structure of the equilibrium remains similar, \emph{ceteris paribus,} agents with lower costs take larger positions and have a larger influence on the equilibrium price.
} satisfying
$$
  0 < \alpha_{-}\leq \alpha_{+},
$$
meaning that the cost of shorting is higher than the cost of going long.  An \emph{admissible} portfolio for an agent is a bounded\footnote{
Boundedness could be replaced by suitable integrability conditions without altering our results.
} process $\Phi$, and we write $\cA$ for the collection of all these portfolios. The value $\Phi(t)$ indicates the number of units of the security held by the agent at time~$t$, and this number can be negative in the case of a short position. Given a (semimartingale) price process $P$, agent $i$ seeks to maximize the expected net payoff\footnote{
To ensure that the expectation is well-defined a priori, we set $E_{i}[Y]:=-\infty$ whenever $E_{i}[\min\{0,Y\}]=-\infty$, for any random variable $Y$. For the processes $P$ that occur in our results below, \eqref{eq:expectedNetPayoff} will be finite for any $\Phi\in\cA$.
}
\begin{equation}\label{eq:expectedNetPayoff}
  E_{i} \left[ \int_{0}^{T} \Phi(t)\,dP(t) - \int_{0}^{T} c(\Phi(t))\,dt\right];
\end{equation}
here the first integral represents the profit-and-loss from trading and the second integral is the cumulative cost-of-carry incurred. Criterion \eqref{eq:expectedNetPayoff} can be rationalized by assuming that agents have access to borrowing and lending at a zero interest rate and that the cost function $c$ is measured in the unit of account but it can also be taken as a primitive utility function. We take  interest rates as exogenous because most bubbles affect only  part of the capital markets and have  limited effect on rates.%\footnote{\cite{Bernanke.10} argues that the federal funds rate did not respond to housing market conditions in 2002--2006 and defends this policy.} 
%For instance, as documented by \cite{KrishnamurthyNagelOrlov.14}, almost 1/4 of the outstanding private-label ABS  were financed with asset backed commercial paper. However, the spread between overnight ABCP and Fed Funds never reached 10~basis points during the bubble. 
The assumption that this exogenous rate equals zero is made to simplify the notation. We account for that in our discussions by referring for instance to the case where $c(y)=0$ for $y\ge 0$ as a constant (rather than zero) marginal cost of being long.
%However, together with the form of the cost function it implies the omission of any first-order asymmetric costs of carrying  positions (but see Appendix~\ref{se:linearCosts} for the model including first-order costs). 
An admissible portfolio $\Phi_{i}$ will be called \emph{optimal} for agents of type $i$ if it maximizes~\eqref{eq:expectedNetPayoff} over all $\Phi\in\cA$. We will examine symmetric equilibria in which  agents of the same type choose the same portfolio. 

As the final input of our model, we introduce a nonnegative \emph{supply function} $s\in C^{1,2}_{b}$.\footnote{Our results could be extended to a discontinuous supply shock as in \cite{hongetal06} using backward induction.} The supply $S(t)=s(t,X(t))$ is owned by third parties that supply the asset inelastically.\footnote{Since the utility function in  \eqref{eq:expectedNetPayoff} is separable, the equilibrium price is invariant to endowments. Hence, we could have alternatively assumed an arbitrary ownership structure for the endowment across the types of investors---we opted for the simpler presentation.} Notice that this formalism allows for the payoff $f(X(T))$ to depend on $S(T).$
An \emph{equilibrium price} is a process\footnote{More precisely, $P$ is a continuous semimartingale, which ensures that the integrals of $\Phi_{i}$ are well-defined. This is automatically satisfied for the processes considered below.} $P$ satisfying $P(T)=f(X(T))$ a.s.\ under all~$Q_{i}$ for which there exist admissible portfolios $\Phi_{i}$, $i\in\{1,\dots,n\}$ such that $\Phi_{i}$ is optimal for agent~$i$ and the market clearing condition
$$
  \sum_{i=1}^{n} \Phi_{i}(t) = S(t)
$$
holds. We are interested in Markovian equilibria; that is, equilibrium prices of the form $P(t)=v(t,X(t))$ for a function $v$ which we refer to as an \emph{equilibrium price function}.

%%%%%%%%%%%%%%%%%%%%%%%%%%%%%%%%%%%%%%%%%%%%
\subsection{Existence and PDE for the Equilibrium Price}

The following notation will be useful to state our first result. Given $v\in C^{1,2}_{b}$, we define the function $\cL^{i}v$ by
\begin{equation}\label{eq:generator}
  \cL^{i}v(t,x) = \partial_{t}v(t,x) + b_i \partial_{x}v(t,x) + \frac12 \tr\sigma^2_i \partial_{xx}v(t,x).
\end{equation}
Here $\partial_{x}v$ denotes the gradient vector, $\partial_{xx}v$ the Hessian matrix, and $\tr\sigma^2_i \partial_{xx}v$ is the trace of the matrix $\sigma^2_i \partial_{xx}v$; that is, the sum of the entries on the diagonal. One can interpret $\cL^{i}v(t,x)$ as the change in $v$ which agents of type $i$ expect over an infinitesimal time interval after~$t$.

Before stating the general characterization of equilibria in Theorem~\ref{th:PDE} below, we develop the heuristics  in two particular cases. We suppose that $X$ is one-dimensional and the coefficients $b_{i}$ and $\sigma_{i}$ are constant.

We first derive the first-order conditions for the portfolios. Suppose that we are in an equilibrium with price $P(t)=v(t,X(t)).$ %and consider an agent of type $i$, thus maximizing the expected net payoff~\eqref{eq:expectedNetPayoff}. 
It\^o's formula states that under $Q_{i}$,
\begin{align*}
  dP(t)
  &=\partial_{t}v(t,x)\,dt + b_i \partial_{x}v(t,x)\,dt + \tfrac12 \sigma_i^2\partial_{xx}v(t,x)\,dt+  \sigma_i\partial_{x}v(t,x)\,dW_{i}(t)\\
  &= \cL^{i}v(t,x)\,dt + \sigma_i\partial_{x}v(t,x)\,dW_{i}(t).
\end{align*}
Thus, the expected final payoff~\eqref{eq:expectedNetPayoff} for a portfolio $\Phi$ is
\begin{align*}
  E_{i} \left[ \int_{0}^{T} \!\!\Phi(t)\,dP(t) - \int_{0}^{T} \!\!c(\Phi(t))\,dt\right]
  &=E_{i}\left[\int_{0}^{T} \!\!\{\Phi(t) \cL^{i}v(t,X(t)) - c(\Phi(t))\}\,dt\right]
\end{align*} 
where we have used that the $dW$-integral has zero expectation. To optimize this quantity, we simply maximize the integrand with respect to $\Phi(t)$ at every~$t$; that is, we set the marginal expected gain $\cL^{i}v(t,X(t))-c'(\Phi(t))=0$. The latter formula shows that the equilibrium \emph{holding premium} for type $i,$ which equals the expected price change $\cL^{i}v$ since the interest rate is zero, is generated by the holding cost.
Using the quadratic form~\eqref{eq:costOfCarry} of $c$, we derive the optimal portfolio
$$
  \Phi_{i}(t)=\phi_{i}(t,X(t)),\quad\mbox{where}\quad \phi_{i}(t,x) = \alpha_{\sgn(\cL^{i}v(t,x))}\cL^{i}v(t,x).
$$
In particular, agents are myopic given the price function and its derivatives (whereas the price itself will incorporate agents' expectations about the future of the state process $X$). %price function $v$ and its derivatives, together with the coefficients $b_{i}$ and $\sigma_{i}$, form a sufficient statistic for the agent's portfolio.

Next, we derive an equation for $v$ in two special cases. 
First, in the homogeneous case where all agents have the same views: $b_{i}=b$ and $\sigma_{i}=\sigma$. Thus, $\cL^{i}v(t,x)=\cL v(t,x)$ is also independent of $i$ and the optimal positions are identical across agents; in particular, there is no short-selling in equilibrium and the first-order condition becomes $\Phi_{i}(t) = \alpha_{+} \cL v(t,X(t))$. Market clearing requires that $\alpha_{+} \cL v(t,X(t))=S(t)/n$ or
$$
  \partial_{t}v(t,x) + b \partial_{x}v(t,x) + \frac12 \sigma^2 \partial_{xx}v(t,x) - \frac{s(t,x)}{n \alpha_+} = 0.
$$
This PDE is linear and supply enters as a running cost: the equilibrium price must compensate for the cost-of-carry.

Second,  consider  $n=2$ types of agents that disagree on the drift coefficient $\mu_{i}$ but agree on the volatility $\sigma:=\sigma_{1}=\sigma_{2}$. To further simplify the derivation, consider the case of zero net supply. Market clearing requires $\phi_{1}+\phi_{2}=0$, thus 
% $$ \alpha_{\sgn(\cL^{1}v(t,x))}\cL^{1}v(t,x) + \alpha_{\sgn(\cL^{2}v(t,x))}\cL^{2}v(t,x) =0.$$
 one type must be long and the other must be short. Therefore, there are two possibilities at every $(t,x)$:
\begin{align*}
  &\cL^{1}v(t,x)\leq0 \mbox{ and } \cL^{2}v(t,x)\geq0,  \mbox{ thus } \alpha_{-}\cL^{1}v(t,x) + \alpha_{+}\cL^{2}v(t,x) =0;\mbox{ or}\\
  &\cL^{1}v(t,x)\geq0 \mbox{ and } \cL^{2}v(t,x)\leq0,  \mbox{ thus } \alpha_{+}\cL^{1}v(t,x) + \alpha_{-}\cL^{2}v(t,x) =0.
\end{align*} 
Recalling that $\alpha_{-}\leq \alpha_{+}$, it  follows that
\begin{align*}
  \mbox{if }&\cL^{1}v(t,x)\leq0 \mbox{ and } \cL^{2}v(t,x)\geq0,  \mbox{ then } \alpha_{+}\cL^{1}v(t,x) + \alpha_{-}\cL^{2}v(t,x) \leq0;\\
  \mbox{if }&\cL^{1}v(t,x)\geq0 \mbox{ and } \cL^{2}v(t,x)\leq0,  \mbox{ then } \alpha_{-}\cL^{1}v(t,x) + \alpha_{+}\cL^{2}v(t,x) \leq0.
\end{align*}
Hence, in all cases:
\begin{align*}
  \max\left\{\alpha_{-}\cL^{1}v(t,x) + \alpha_{+}\cL^{2}v(t,x) ;\,\alpha_{+}\cL^{1}v(t,x) + \alpha_{-}\cL^{2}v(t,x)\right\}=0.
\end{align*}
Next, divide the above equation by $(\alpha_{-}+\alpha_{+})$ and plug in the definitions of $\cL^{1}$ and $\cL^{2}$.
After rearranging  terms, one obtains
$$
  \partial_{t}v(t,x) + \!\!\max_{(i,j)=(1,2),(2,1)}\!\left\{\left(\tfrac{\alpha_{-}}{\alpha_{-}+\alpha_{+}}b_i +  \tfrac{\alpha_{+}}{\alpha_{-}+\alpha_{+}}b_j\right)\!\partial_{x}v(t,x)\right\} + \frac12 \sigma^2\partial_{xx}v(t,x)=0.
$$
Disagreement about drifts  caused a non-linearity in the first-order term. Similarly, disagreement about volatilities would have caused a non-linearity in the second-order term.
The next theorem  states that an analogous PDE uniquely characterizes the equilibrium price function in our model. In general, the above maximization over two possibilities is replaced by a maximization over ``groups'' $I\subseteq \{1,\dots,n\}$ of agents; we denote by $|I|$ the number of agents in $I$ and by $I^{c}=\{1,\dots,n\}\setminus I$ the complementary group.

\begin{theorem}\label{th:PDE}
  (i) There exists a unique equilibrium price function $v\in C^{1,2}_{b}$. The corresponding optimal portfolios are unique\footnote{Uniqueness is understood up to $(Q_{i}\times dt)$-nullsets.} and given by $\Phi_{i}(t)=\phi_{i}(t,X(t))$, where 
  \begin{equation}\label{eq:optPortfolioMain}
    \phi_{i}(t,x) = \alpha_{\sgn(\cL^{i}v(t,x))}\cL^{i}v(t,x).
  \end{equation}
  
  (ii) The function $v\in C^{1,2}_{b}$ can be characterized as the unique solution of the PDE
  \begin{equation}\label{eq:mainPDE}
  \partial_{t}v(t,x) + \sup_{I\subseteq \{1,\dots,n\}} \Big(\mu_{I}(t,x)\partial_{x}v(t,x) + \tfrac12\tr\Sigma_{I}^{2}(t,x)\partial_{xx}v(t,x) - \kappa_{I}(t,x)\Big) = 0
  \end{equation}
  on $[0,T)\times\R^{d}$ with terminal condition $v(T,x)=f(x)$, where the supremum is taken over all subsets $I\subseteq \{1,\dots,n\}$ and the coefficients are defined as
  %\pagebreak[1]
  \begin{equation}\label{eq:defMu}
      \mu_{I}(t,x) = \tfrac{\alpha_{-}}{|I|\alpha_{-} + |I^{c}|\alpha_{+}} \sum_{i\in I} b_{i}(t,x) + \tfrac{\alpha_{+}}{|I|\alpha_{-} + |I^{c}|\alpha_{+}}\sum_{i\in I^{c}} b_{i}(t,x),
  \end{equation}
    \begin{equation}\label{eq:defSigma}
      \Sigma_{I}^{2}(t,x) = \tfrac{\alpha_{-}}{|I|\alpha_{-} + |I^{c}|\alpha_{+}} \sum_{i\in I} \sigma_{i}^{2}(t,x) + \tfrac{\alpha_{+}}{|I|\alpha_{-} + |I^{c}|\alpha_{+}}\sum_{i\in I^{c}} \sigma_{i}^{2}(t,x),
  \end{equation}
  \begin{equation}\label{eq:defCost} 
  \kappa_{I}(t,x) = \frac{s(t,x)}{|I|\alpha_{-} + |I^{c}|\alpha_{+}}.
  \end{equation}  
  Moreover, a maximizer for the supremum in~\eqref{eq:mainPDE} is given by 
  \begin{equation}\label{eq:optControlThm}
    I_{*}(t,x) = \{i\in \{1,\dots,n\}:\, \cL^{i}v(t,x)<0\}.
  \end{equation}
\end{theorem}

In equilibrium, the group $I_{*}$ of~\eqref{eq:optControlThm} corresponds to the more pessimistic agents (holding shorts) whereas $I_{*}^{c}$ are the optimists (holding long positions).
The formulas~\eqref{eq:defMu} and~\eqref{eq:defSigma} for $\mu_{I}$ and $\Sigma_{I}$ can be seen as a weighted average of the drift and volatility coefficients of the agents. The weights entail that when shorting is  more expensive than being long (i.e., $\alpha_-$ is small relative to $\alpha_+$),  optimists have a larger impact on the equilibrium price. In Section~\ref{se:comparativeStatics} we show that when $\alpha_-\to 0$ or $\alpha_+\to \infty$, the more pessimistic views are not reflected in the equilibrium price at all. The running cost $\kappa_{I}$ of~\eqref{eq:defCost} depends linearly on the exogenous supply $s$ which is divided by a weighted sum of the cost coefficients, the weights  being the size of the set $I$ and its complement $I^{c}$, respectively. Since $\alpha_{-}\leq \alpha_{+}$, the cost increases with the number~$|I|$ of types in the group $I$. 

We will see in Section~\ref{se:planner} that the precise form of the PDE~\eqref{eq:mainPDE} with a supremum can be explained through the problem of a planner with limited instruments. To obtain an initial intuition, note that using~\eqref{eq:generator}, the left-hand side of the PDE can be read as the difference of two quantities. The first one is a weighted average over the instantaneous holding premia. The second quantity is related to the instantaneous marginal cost of carrying positions. The PDE equates this difference to zero when the weights correspond to the particular group $I=I_{*}$.

Mathematically, the PDE~\eqref{eq:mainPDE} is of Hamilton--Jacobi--Bellman type, which entails that~$v$ can be represented as the value function of a stochastic optimal control problem. This is useful for our derivation of  comparative statics and limiting results presented below but has no obvious economic interpretation. The control problem is detailed in Appendix~\ref{se:controlRepresentation}.

\begin{remark}\label{rk:scaling}
  The equilibrium price $v(t,x)$ is 0-homogeneous in $(\alpha_{-},\alpha_{+},s)$, indicating that supply and costs are closely linked in our model. That is, if these parameters are replaced by $(\lambda\alpha_{-},\lambda\alpha_{+},\lambda s)$ for some $\lambda>0$, the price does not change. This follows from Theorem~\ref{th:PDE}\,(ii) after observing that the coefficients $\mu_{I}$, $\Sigma_{I}$ and $\kappa_{I}$ are invariant under this substitution. %For instance, increasing the supply by a factor $\lambda$ has the same effect on the price as increasing both costs-of-carry by the same factor: $(\alpha_{-},\alpha_{+}, \lambda s)$ and  $(\lambda^{-1}\alpha_{-},\lambda^{-1}\alpha_{+},s)$ produce the same equilibrium price. 
In the special case $s=0$, the homogeneity entails that the price depends on $(\alpha_{-},\alpha_{+})$ only through the ratio $\alpha_{+}/\alpha_{-}$.
\end{remark}

%%%%%%%%%%%%%%%%%%%%%%%%%%%%%%%%%%%%%%%%%%%%%%%%
\section{Comparative Statics and Limiting Cases}\label{se:comparativeStatics}

In the first part of this section we establish comparative statics with respect to the supply and cost parameters.  In the second part we analyze the limit $\alpha_{+}\to \infty$ when there is no cost-of-carry for long positions, as well as the limit $\alpha_{-}\to 0$ when short positions are ruled out.

%%%%%%%%%%%%%%%%%%%%%%%
\subsection{Comparative Statics}

We start with the dependence on the supply.

\begin{proposition}\label{pr:monotonicitySupply}
  The equilibrium price function $v$ is monotone decreasing with respect to the supply function $s$: prices decrease with an increase in supply.
\end{proposition}

Next, we turn to the cost parameters $\alpha_{-}$ and $\alpha_{+}$. The following shows that the equilibrium price is decreasing with respect to the cost-of-carry for long positions and increasing with respect to the cost for short positions.

\begin{proposition}\label{pr:monotonicityCost}
  The equilibrium price function $v$ is 
  \begin{enumerate}
  \item increasing with respect to $\alpha_{+}$,
  \item decreasing with respect to $\alpha_{-}$,
  \item increasing with respect to the quotient $\alpha_{+}/\alpha_{-}$ if $s\equiv0$.
  \end{enumerate}
\end{proposition}

The proof uses our PDE characterization of the price and a comparison theorem from the theory of parabolic partial differential equations.\footnote{Comparison theorems are useful to show that two functions satisfy an inequality on their domain if they are known to satisfy an (in)equality on the boundary. In our context, we may think of the equilibrium price function as satisfying a PDE $F(v,\beta)=0$ where $\beta$ is a parameter. If $v_{1}$ and $v_{2}$ are price functions corresponding to different parameters $\beta_{1}$ and $\beta_{2}$, we know that they are equal at the boundary $t=T$ since they satisfy the same terminal condition $f$. If $v_{2}$ is a subsolution of the PDE for $v_{1}$, that is $F(v_{2},\beta_{1})\geq0$, the comparison theorem implies that $v_{1}\geq v_{2}$. See e.g.\ \cite{FlemingSoner.06}.}

%we compare equilibrium price functions $v_{1}$ and $v_{2}$ which are solutions of PDEs of the form $F(v_{j},\beta_{j})=0$ where $\beta_{j}$ represents a vector of parameters. We known that both functions satisfy the boundary condition $v_{j}(T,\cdot)=f$. If $v_{2}$ is a subsolution of the PDE for $v_{1}$, that is $F(v_{2},\beta_{1})\geq0$, the comparison theorem implies that $v_{1}\geq v_{2}$.}

% state that if two functions satisfy an inequality $v_{1}\leq v_{2}$ on the boundary of their domain and solve suitable differential inequalities inside the domain, then the inequality $v_{1}\leq v_{2}$ holds everywhere on the domain. Here, we are interested in comparing functions $v_{1}$ and $v_{2}$ which are equal at the boundary $t=T$ but solve different PDEs (corresponding to different parameters $\alpha_{\pm}$). In the Appendix, we use these two PDEs to construct an auxiliary differential inequality to which the comparison theorem can be applied.}
%\footnote{{I would write instead: We may think of the equilibrium price function as satisfying an equation $F(v,\beta)(t,x)=0$ for each $(t,x)$  where $\beta$ is a parameter.Let $v_i$ be a solution for parameter value $\beta_i.$ A comparison theorem  states that if  $F(v_2, \beta_1)(t,x) \le 0$ then $ v_1(t,x) \le  v_2(t,x).$ THEN IN THE PROOF IN THE APPENDIX I WOULD SWITCH SIGNS SO THAT THE NORMAL SIGN CONVENTION APPLIES}}

The following is a partial extension of~(iii) to the case of non-zero supply which is useful if $\alpha_{-}$ and $\alpha_{+}$ are varied simultaneously.

\begin{remark}\label{rk:preciseComparison}
   Let $\alpha_{-}\leq\alpha_{+}$ and $\alpha'_{-}\leq\alpha'_{+}$ be two pairs of cost coefficients and let $v$ and $v'$ be the corresponding equilibrium price functions. 
  If the coefficients satisfy $\alpha_{+}/\alpha_{-}\leq \alpha'_{+}/\alpha'_{-}$ and $\alpha_{-}\leq\alpha'_{-}$, then $v\leq v'$.
  
  For instance, it follows that if the costs-of-carry for long and short positions are increased by a common factor, then the price decreases.
\end{remark}

%%%%%%%%%%%%%%%%%%%%%%%
\subsection{Limiting Models}\label{se:limitingModels}

We discuss two limits for the cost coefficients which help to understand the relationship between our model and the earlier ones  discussed in the Introduction.
To make the dependence on the parameters explicit, we denote by $v^{\alpha_{-},\alpha_{+}}(t,x)$ the equilibrium price function $v(t,x)$ for  $\alpha_{-},\alpha_{+}$.

%%%%%%%%%%%%%%%%%%%%%%%
\subsubsection{Zero Cost for Long Positions}

We first consider the limit $\alpha_{+}\to \infty$ when the cost-of-carry for long positions tends to zero.

\begin{proposition}\label{pr:limitLong}
  As $\alpha_{+}\to \infty$, the function $v^{\alpha_{-},\alpha_{+}}$ converges to the unique solution $v^{\infty}\in C^{1,2}_{b}$ of the PDE
  \begin{equation}\label{eq:PDEalphaPlusInfty}
  \partial_{t}v + \sup_{i\in \{1,\dots,n\}} \left( b_{i}\partial_{x}v + \tfrac12\tr \sigma_{i}^{2}\partial_{xx}v \right) = 0
  \end{equation}
  with terminal condition $v(T,x)=f(x)$; in particular,  $v^{\infty}$ is independent of~$\alpha_{-}$ and $s$. The convergence is locally uniform in $(t,x)$, and monotone increasing  as $\alpha_{+}\uparrow \infty$.
\end{proposition}

We now discuss the limiting model that arises in Proposition~\ref{pr:limitLong}; that is, with no cost-of-carry for long positions. We state these results without proofs since these would be very similar to the proof of Theorem~\ref{th:PDE}.

The limiting model has an equilibrium price function $v:=v^{\infty}$ that is unique and independent of the supply $s$ and the cost coefficient $\alpha_{-}$ for short positions. Thus, we retrieve the results of previous models with risk-neutral agents in this limiting regime. The intuition for Equation~\eqref{eq:PDEalphaPlusInfty} is straightforward. In any equilibrium, if $j$ is one of the most optimistic types, we  must have $\cL^{j}v(t,x)\ge 0.$ However if the marginal cost of going long is zero, $\cL^{j}v(t,x)=0$ necessarily holds. In particular, $j$ is indifferent with respect to nonnegative positions and  equilibrium prices are independent of the supply for the asset or the demand for shorting.  
However, the optimal portfolios $\Phi_{i}(t)=\phi_{i}(t,X(t))$ in equilibrium do depend on $s$ and $\alpha_{-}$. Given $(t,x)$, if $i$ is not a maximizer, $\cL^{i}v(t,x)<0$ and 
$$
\phi_{i}(t,x)=\alpha_{-} \cL^{i}v(t,x)
$$
as in~\eqref{eq:optPortfolioMain}; in particular, agent $i$ holds a short position. In equilibrium, the aggregate amount held by the most optimistic types is set by the market clearing condition---they must hold  the sum of the exogenous supply and all amounts shorted.
%$$
% s(t,x) - \sum_{j\in\{1,\dots,n\}\setminus\{i\}} \alpha_{-}\cL^{j}v(t,x)
%$$
%This explains why the price does not depend on $s$ and $\alpha_{-}$; in fact,  price is solely determined by the characteristics of the most optimistic types. 
If there is more than one maximizer $i$, then any distribution of the available amount (supply plus short positions) over these maximizers gives an optimal allocation.\footnote{See \cite{MuhleKarbeNutz.16} for an analysis of this case when shorting is constrained.}

The properties described above for $\alpha_+=\infty$ continue to hold if in addition  $\alpha_{-}=0$; i.e., when there is no cost for long positions and short positions are prohibited. In particular, all but the most optimistic agents hold a flat position, and only the most optimistic characteristics play a role in determining the price. Thus, we retrieve the results of previous models with risk-neutral agents in this limiting regime.

\begin{remark}\label{rk:symmetricCost}
 The results for $\alpha_{+}=\infty$ may be contrasted with the opposite extreme case where the cost coefficients $\alpha_{+}$ and $\alpha_{-}$ are equal. Then, the drift and volatility coefficients
  $$
    \mu:=\mu_{I} = \frac1n \sum_{i=1}^{n} b_{i},\qquad \Sigma^{2}:=\Sigma_{I}^{2} = \frac1n \sum_{i=1}^{n} \sigma_{i}^{2}
  $$
  are independent of $I$ and equal to the arithmetic average of the coefficients in the agents' models, meaning that all agents contribute equally to the price. The running cost is $\kappa:=\kappa_{I}=s/(n\alpha_{+})$. Thus, \eqref{eq:mainPDE} becomes the linear PDE
  $$
  \partial_{t}v + \frac1n \sum_{i=1}^{n} b_{i}\partial_{x}v(t,x) + \frac{1}{2n} \sum_{i=1}^{n} \tr\sigma_{i}^{2}\partial_{xx}v - \frac{s}{n\alpha_{+}} = 0
  $$
   and then by the Feynman--Kac formula, the equilibrium price is
  $$%\label{eq:symmetricCostPrice}
    v(t,x) = E\left[f(X^{t,x}(T)) - \int_{t}^{T} s(r,X^{t,x}(r))/(n\alpha_{+})  \,dr\right]
  $$
  where $X^{t,x}$ is a diffusion with drift $\mu$, volatility $\Sigma$ and initial condition $X^{t,x}(t)=x$ (cf.\ Appendix~\ref{se:controlRepresentation}). That is, the equilibrium price is simply the expected value of the security under the averaged coefficients of the agents, minus a  cost term related to the supply.
\end{remark}

%%%%%%%%%%%%%%%%%%%%%%%
\subsubsection{Infinite Cost for Short Positions}\label{se:infiniteCostShort}

We now discuss the limit $\alpha_{-}\to 0$; that is, the cost-of-carry for short positions tends to infinity.

\begin{proposition}\label{pr:limitShort}
  As $\alpha_{-}\to 0$, the function $v^{\alpha_{-},\alpha_{+}}$ converges to the unique solution $v^{0,\alpha_{+}}\in C^{1,2}_{b}$ of the PDE
  \begin{equation}\label{eq:PDEalphaMinusZero}
  \partial_{t}v + \sup_{\emptyset\neq J\subseteq \{1,\dots,n\}} \left(\tfrac{1}{|J|}\sum_{i\in J} b_{i}\partial_{x}v + \tfrac12\tr \tfrac{1}{|J|}\sum_{i\in J} \sigma_{i}^{2}\partial_{xx}v - \frac{s}{|J|\alpha_{+}}\right) = 0
  \end{equation}
  with terminal condition $v(T,x)=f(x)$. In the special case where $s=0$, this PDE coincides with~\eqref{eq:PDEalphaPlusInfty} and in particular the solution $v^{0,\alpha_{+}}=v^{\infty}$ is independent of $\alpha_{+}$. The convergence is locally uniform in $(t,x)$, and monotone increasing if $\alpha_{-}\downarrow 0$.
\end{proposition}

The limiting model that arises in Proposition~\ref{pr:limitShort} corresponds to a prohibition of shorting. This model has a unique equilibrium price function $v:=v^{0,\alpha_{+}}$ which depends on the supply $s$ and the cost coefficient $\alpha_{+}$ for long positions.
At every state $(t,x)$, we can think of the types as being divided into a group 
$J=\{i\in\{1,\dots,n\}:\, \cL^{i}v(t,x)\geq0\}$ of relatively optimistic agents and the complement $J^{c}$ of 
pessimists. We have $J\neq\emptyset$ by market clearing. While the agents in~$J$ hold positions $\alpha_{+} \cL^{i}v(t,x)$ of different magnitude depending on how optimistic they are, the entire group $J$ determines the price. Agents in $J^{c}$, however, hold zero units and their precise characteristics do not enter the formation of the price. For instance, if we replace a pessimistic type $i\in J^{c}$ by an even more pessimistic type, the equilibrium price will not change.

%We remark that the case $s=0$ is degenerate in the limiting model: while the above assertions are valid, the portfolios are in fact given by $\Phi_{i}\equiv0$ for all $i$. Indeed, the PDE implies $\cL^{i}v(t,x)\leq0$. Alternately, this follows from the fact that there is neither supply nor shorting in this case.

%%%%%%%%%%%%%%%%%%%%%%%%%%%%

\section{Speculation}\label{se:speculation}

In this section we highlight the impact of non-linear costs-of-carry and short-selling on the pricing mechanism by comparing the above ``dynamic'' equilibrium price at time $t=0$ with a ``static'' equilibrium price; that is, an equilibrium without speculation. We shall see that, as in previous models, the dynamic price dominates the static price when cost-of-carry and short-selling are removed from our model. This can be attributed to the resale option. However, we  show that the cost-of-carry (i.e., risk aversion) gives rise to a  delay option that may act in opposition to the resale option and reverse the order of the prices in extreme cases---even if short-selling is prohibited. Moreover, we 
 illustrate that the possibility of short-selling tends to depress the dynamic price as it gives rise to a repurchase option for pessimists.

%%%%%%%%%%%%%%%%%%%%%%%%%%%%
\subsection{Equilibrium without Speculation}\label{se:SE}

Consider a situation where trading occurs only at the initial time $t=0$; that is,  agents are forced to use buy-and-hold strategies and speculation is ruled out.
The agents use the same models $Q_{i}$ for the dynamics~\eqref{eq:SDEforX} of the state process~$X$ and maximize the same expected net payoff~\eqref{eq:expectedNetPayoff}. However, the admissible portfolios $\Phi$ are restricted to be constant, and we will use the letter $q$ to denote a generic portfolio. This market can only clear if the exogenous supply $S\equiv s$ is constant, so we restrict our attention to that case. A \emph{static equilibrium price} is defined like the dynamic equilibrium price above, except that we only look for a constant $p_{\sta}\in\R$ at time $t=0$ at which the trading happens.

\begin{proposition}\label{pr:staticPrice}
  (i) There exists a unique static equilibrium price and it is given by
  \begin{equation}\label{eq:staticPrice}
    p_{\sta}= \max_{I\subseteq\{1,\dots,n\}}  \left(
  \tfrac{\alpha_{-}}{|I|\alpha_{-} + |I^{c}|\alpha_{+}} \sum_{i\in I} e_{i} 
  + \tfrac{\alpha_{+}}{|I|\alpha_{-} + |I^{c}|\alpha_{+}} \sum_{i\in I^{c}} e_{i}
  - \tfrac{sT}{|I|\alpha_{-} + |I^{c}|\alpha_{+}}\right),
  \end{equation}
  where $e_{i}=E_{i}[f(X(T))]$. The corresponding optimal static portfolios are unique and given by 
  \begin{equation}\label{eq:staticPortfolio}
    q_{i} = \alpha_{\sgn(e_{i}-p_{\sta})} T^{-1} (e_{i}-p_{\sta}).
  \end{equation}
\end{proposition}

The formula for the static price is the direct analogue of the PDE~\eqref{eq:mainPDE} for the dynamic price. 
Indeed, the PDE considers the difference between a weighted average of instantaneous holding premia and the instantaneous marginal cost of carrying positions. Formula~\eqref{eq:staticPrice} can be read in the same way, after bringing $p_{\sta}$ to the right-hand side: it considers the difference between the weighted average of the  holding premia $e_{i}-p_{\sta}=E_{i}[f(X(T))]-p_{\sta}$ over the whole interval and the marginal cost of carrying positions over that same interval. Informally, we may think of the PDE as describing a repeated version of the static problem over infinitesimal intervals.
%
%Indeed, in the latter we take a weighted average of the instantaneous expected prince increments $\cL^{i}v$ and set it to zero (after adjusting for holding costs) to describe that agents are invariant in equilibrium. The formula~\eqref{eq:staticPrice} can be read in the same way after bringing $p_{\sta}$ to the right-hand side: the weighted average of the expected prince increments $e_{i}-p_{\sta}=E_{i}[f(X(T))]-p_{\sta}$, now over the whole time interval rather than instantaneous, is set to zero. 
%
%Observe that the formula~\eqref{eq:staticPrice} for the static price is reminiscent of the PDE~\eqref{eq:mainPDE} for the dynamic price; however, the averaging of the models of the different types now occurs at the level of the expected values $e_{i}=E_{i}[f(X(T))]$ instead of the instantaneous drift and volatility  $b_{i},\sigma_{i}$. Informally, we may think of the PDE as describing a repeated version of the static problem over infinitesimal intervals.

In parallel with our analysis above, we can consider limiting cases for the cost coefficients in the static case. 
We denote by $p_{\sta}^{\alpha_{-},\alpha_{+}}$ the static equilibrium price for cost parameters $\alpha_{-},\alpha_{+}$ and initial value $X(0)=x$ as given by~\eqref{eq:staticPrice}.

\begin{proposition}\label{pr:limitStatic}
  (i) In the limit $\alpha_{+}\to \infty$ with zero cost for long positions, the price $p_{\sta}^{\alpha_{-},\alpha_{+}}$ converges to
  \begin{equation}\label{eq:StaticPriceAlphaPlusInfty}
  p_{\sta}^{\infty} = \max_{i\in\{1,\dots,n\}} E_{i}[f(X(T))].
  \end{equation}
  (ii) In the limit $\alpha_{-}\to 0$ with infinite cost for short positions, the price $p_{\sta}^{\alpha_{-},\alpha_{+}}$ converges to
  \begin{equation}\label{eq:StaticPriceAlphaMinusZero}
    p_{\sta}^{0,\alpha_{+}} = \max_{\emptyset\neq J \subseteq \{1,\dots,n\}} \left(\tfrac{1}{|J|} \sum_{i\in J} E_{i}[f(X(T))]
  - \frac{sT}{|J|\alpha_{+}}\right).
  \end{equation}
\end{proposition}

The intuition is the same as in Section~\ref{se:limitingModels}. Without a cost for holding long positions, optimists are indifferent with respect to nonnegative portfolios and the price is solely determined by the most optimistic agents. Whereas when shorting is ruled out, the price is determined as an average over a group of relatively more optimistic agents while the complementary group of more pessimistic agents does not influence the price directly. We omit the  proof of Proposition~\ref{eq:StaticPriceAlphaPlusInfty} since it  is analogous to  Section~\ref{se:limitingModels}. 

%%%%%%%%%%%%%%%%%%%%%%%%%%
\subsection{Resale and Delay Options}\label{subsec:randdo}

Next, we compare the dynamic equilibrium price $p_{\dyn}:=P(0)$ at time $t=0$ with the static equilibrium price $p_{\sta}$. For the latter to be well-defined, we assume throughout that the supply $s$ is constant. We discuss  two options  that are present under dynamic trading and are valued by  agents---the resale and delay options---and the effect on prices of eliminating these options by forcing agents to trade only at time zero. In particular, we shall see that the ordering of $p_{\dyn}$ and $p_{\sta}$ may be different than in the earlier models.

Previous papers, starting with \cite{harrisonkreps78},   considered models with risk-neutral agents that face a constant marginal cost-of-carry for long positions (the interest rate)  and cannot  sell short.  In such models, it is known that the dynamic equilibrium price exceeds the static one, and the difference is attributed to the ``resale option.'' The possibility of reselling the asset increases the price---agents may want to buy today in order to resell to agents that are more optimistic tomorrow. In these ``classical'' models, agents may also plan to buy additional units of the asset in some future states of the world.   This possibility however does not alter the ranking between the dynamic and static equilibrium prices. Indeed, since agents are risk-neutral and the marginal cost of carrying a long position is independent of the size of the position, we may assume generically that only one type $i$ would acquire the asset in the static equilibrium and pay its marginal valuation at time zero. When re-trading is allowed, $i$'s marginal valuation  for holding the full supply of the asset at time zero is at least as large, since an agent can always choose a buy-and-hold strategy. As  the market price must exceed the marginal valuation of any type, the dynamic equilibrium price must exceed the static equilibrium price.
  
The next two results confirm this intuition by showing how this mechanism carries over to limiting cases of our model. First, we show that when the marginal cost of long positions is constant, the dynamic price exceeds the static one. This holds even when shorting is allowed, because in this extreme case, only the most optimistic agents contribute to the price formation, just like in the classical models (see also Propositions~\ref{pr:limitLong} and~\ref{pr:limitStatic}).

\begin{proposition}\label{pr:limitLongComparison}
  In the limit $\alpha_{+}\to \infty$, the dynamic equilibrium price dominates the static price:
  $
    p_{\dyn}^{\infty} \geq p_{\sta}^{\infty}.
  $
\end{proposition}

Next, we show that if short-sales are prohibited \emph{and} if in the static equilibrium only one type holds the asset,\footnote{Only one type will hold the asset if that type is sufficiently more optimistic than the others and the supply is small enough.}
 the dynamic equilibrium price again exceeds the static price, even when longs face an increasing marginal cost-of-carry. 

\begin{proposition}\label{pr:comparisonOneOptimist}
In the limit $\alpha_{-}\to 0$ with no short-selling, suppose type $i$ holds the entire market in the static equilibrium; that is, $q_{j}=0$ for $j\neq i$. Then, the dynamic equilibrium price dominates the static price:
$
p_{\dyn}^{0,\alpha_{+}} \geq p_{\sta}^{0,\alpha_{+}}.
$
\end{proposition}

We now turn to the case when both marginals costs are increasing and finite. Here, the same options to resell and to delay are present, but the effects are more subtle.  The option to delay now affects  equilibrium prices because the marginal valuation of buyers varies with the size of their position. More importantly, trading may occur in the dynamic equilibrium  even  though one type remains the most optimistic. Indeed, in the classical models (and the limiting model of Proposition~\ref{pr:limitLong}) the most optimistic type always holds the full supply and trading requires that relative optimism changes sign.  When  the marginal cost-of-carry for long positions is increasing, the magnitude of relative optimism determines  equilibrium holdings---it is no longer true that a less optimistic type would always hold a non-positive amount. Example~\ref{ex:negativeBubble} below illustrates that the delay option may have an important impact on prices and even reverse the ordering of dynamic and static prices.

If shorting is allowed, buying today in order to resell needs to be compared with entering a short position tomorrow. The choice  will depend, among other factors, on the costs-of-carry for long and short positions. The option to resell a short position---that is, the option to cover shorts---in turn, decreases the minimum amount pessimists would be willing to receive for shorting the asset, putting downward pressure on prices. Shorts also may exercise the option to delay by building up a short position over time. Example~\ref{ex:symmetricCosts} below illustrates how the ordering of dynamic and static prices can be reversed if the cost of shorting is sufficiently low.

\begin{remark}\label{rk:violatingAssumptions}
In the remainder of this section, we use a quadratic payoff function $f$ to obtain explicit formulas. This violates our assumption that $f$ is bounded but our results still apply with the appropriate modifications; in particular, the equilibrium price function $v$ and the admissible portfolios $\phi_{i}$ exhibit polynomial growth instead of being bounded. The formulas in our examples can also be verified by direct calculation.
\end{remark}

%%%%%%%%%%%%%%%%%%%%%%%%%%
\subsection{Illustrating the Effect of the Delay Option}

In this section, we show that the static price may dominate the dynamic price even when short-selling is prohibited. This cannot be explained with a resale option; instead, it highlights the delay option.
Consider first the dynamic equilibrium and suppose that type $i$ expects with high probability that their portfolio $\Phi_{i}(t)$ will increase over time.
If only buy-and-hold strategies are allowed, an agent of type~$i$ would consider \emph{anticipating} the increase of the portfolio at time $t=0$, and if the additional expected gains outweigh the additional costs-of-carry, the agent would have a higher buy-and-hold demand at the previous equilibrium price $p_{\dyn}$. Other types may reduce their positions at the price $p_{\dyn}$, because they are anticipating a decrease in position or because they are indifferent to the amount they are holding (see also Example~\ref{ex:noCost} in the Appendix).

To show that the static price may exceed the dynamic price, even when short-selling is prohibited ($\alpha_{-}=0$), we impose a positive cost for long positions ($\alpha_{+}=1$) and construct an example where some agents expect to increase their positions over time but no agent expects to decrease their position.\footnote{Since types disagree, it may indeed happen that all agents expect to increase their positions over time in the dynamic case, without contradicting the market clearing condition.} To obtain explicit formulas despite the nonlinear context, we consider the limiting case of zero volatility but show later (Proposition~\ref{pr:limitSmallVolatilityEx}) that this is indeed the continuous limit for equilibria with small volatility coefficients~$\sigma_{i}$. In particular, the qualitative conclusions of the example extend to examples with diffusion risk. The zero-volatility case violates our assumption of uniformly parabolic coefficients  (and indeed $v$ is not smooth in this example) but the formulas can be verified by direct calculation.

\begin{example}\label{ex:negativeBubble}
  Consider $n=2$ types with volatility coefficients $\sigma_{i}=0$ and constant, opposing drifts
  $$
    b_{1}=1,\qquad b_{2}=-1.
  $$  
  The payoff function is $f(y)=y^{2}$ and the supply $s>0$ is constant. Moreover, $\alpha_{-}=0$ and $\alpha_{+}=1$.
  Then, as we show in Appendix~\ref{se:furtherEx}, the static equilibrium price exceeds the dynamic price; more precisely,
  $$
  p_{\sta}-p_{\dyn} = \begin{cases}
   T^{2}, & |x|\leq s/4-T/2,\\
   (s/2 - 2|x|)T, &  s/4-T/2 < |x| < s/4,\\
  0, & |x|\geq s/4.
  \end{cases}
  $$
\end{example}

In the regimes $s/4-T/2 < |x| < s/4$ and $|x|\leq s/4-T/2$, at least one of the types has a dynamic portfolio that is increasing in time. These agents are exercising the delay option when re-trading is allowed and have an anticipatory motive when they can only trade at $t=0$. A price increase is necessary to clear the static market, leading to $p_{\sta}>p_{\dyn}$. In Appendix~\ref{se:furtherEx} we discuss in detail the asset allocation in all regimes and show how the delay option explains the difference $p_{\sta}-p_{\dyn}$.

It remains to prove that the conclusions of the example also hold when volatilities are small but positive, rather than vanishing.

\begin{proposition}\label{pr:limitSmallVolatilityEx}
  Consider the setting of Example~\ref{ex:negativeBubble} with constant volatilities $\sigma:=\sigma_{1}=\sigma_{2}\geq0$ and denote the corresponding static and dynamic equilibrium prices by $p_{\sta}^{\sigma}$ and $p_{\dyn}^{\sigma}$, respectively. Then, $p_{\sta}^{\sigma}\downarrow p_{\sta}^{0}$ and $p_{\dyn}^{\sigma}\downarrow p_{\dyn}^{0}$ as $\sigma\downarrow 0$. As a consequence, we have
  $$
    p_{\sta}^{\sigma} - p_{\dyn}^{\sigma} \to 
   \begin{cases}
   T^{2}, & |x|\leq s/4-T/2,\\
   (s/2 - 2|x|)T, &  s/4-T/2 < |x| < s/4,\\
  0, & |x|\geq s/4.
  \end{cases}
  $$
\end{proposition}

The above example of the delay option effect should be contrasted with Proposition~\ref{pr:limitLongComparison} where we have seen that when there is no cost-of-carry for long positions ($\alpha_{+}=\infty$), the dynamic equilibrium price always exceeds the static one, even if short-selling is possible.
Example~\ref{ex:noCost} illustrates the mechanics of the delay option in the latter situation. In Example~\ref{ex:noCost}, pessimists plan to close their short position over time in the dynamic equilibrium. When forced to buy-and-hold, they decrease their initial short position; however, in contrast to Example~\ref{ex:negativeBubble},  this has no effect on the static price because, as we have argued,  optimists are indifferent to the size of their own position in the absence of increasing marginal costs.

%%%%%%%%%%%%%%%%%%%%%%%%%%

\subsection{Illustrating the Effect of Shorting}

The following example illustrates that when shorting is allowed, the static price may exceed the dynamic price---this is quite natural once we observe the symmetry between optimists and pessimists in the extreme case $\alpha_{-}=\alpha_{+}$. The difference between the dynamic price and the static price has been identified as the size of the ``speculative bubble'' in the previous literature. If we maintain this identification, the example can be used to illustrate how lowering the cost of shorting can lead not only to a bubble implosion but even to a negative bubble. 

\begin{example}\label{ex:symmetricCosts}
To facilitate computations, we assume symmetric costs-of-carry $\alpha_{-}=\alpha_{+}=1.$ Consider $n=2$ types with constant coefficients $b_{i}\in\R$ and $\sigma_{i}>0$, and an asset in zero aggregate supply with payoff $f(y)=y^{2}$. Writing $\Sigma^{2}:=(\sigma_{1}^{2}+\sigma_{2}^{2})/2$ and $\mu:=(b_{1}+b_{2})/2$, the dynamic and static equilibrium prices at $t=0$ for the initial value $X(0)=x$ are
\begin{align*}
p_{\dyn} &= x^{2} + 2x\mu T + \Sigma^{2}T + \left(\frac{b_{1}+b_{2}}{2}\right)^{2}T^{2},\\
p_{\sta} &= x^{2} + 2x\mu T + \Sigma^{2}T + \frac{b_{1}^{2}+b_{2}^{2}}{2} T^{2};
\end{align*}
see Appendix~\ref{se:proofs} for the calculations. In particular, 
$$
p_{\dyn} - p_{\sta} = \left[\left(\frac{b_{1}+b_{2}}{2}\right)^{2} - \frac{b_{1}^{2}+b_{2}^{2}}{2}\right] T^{2} \leq0.
$$
The optimal dynamic and static portfolios are given by
\begin{align*}
\phi_{i}(t,x) &= x(b_{i}-b_{j}) + \tfrac12 (T-t)(b_{i}^{2}-b_{j}^{2}) + \tfrac12 (\sigma_{i}^{2}-\sigma_{j}^{2}),\\
q_{i} &= x(b_{i}-b_{j}) + \tfrac12 T(b_{i}^{2}-b_{j}^{2}) + \tfrac12 (\sigma_{i}^{2}-\sigma_{j}^{2}),
\end{align*}
where $j=2$ if $i=1$ and vice versa; in particular, the demands at $t=0$ coincide.
In the special case $b_{1}=b_{2}$ where all agents agree on the drift, we have $p_{\dyn}=p_{\sta}$ and the demands coincide at all times. Whenever $b_{1}\not =b_{2},$ a continuity result similar to the results established in Section \ref{se:limitingModels} guarantees that $p_{\dyn}<p_{\sta}$ for cost parameters close to $\alpha_{-}=\alpha_{+}=1$.
\end{example}

To obtain some intuition for this example, consider the case were $\sigma_1=\sigma_2,$ $b_1>0$ and $b_2=0.$ Then if $x>0$, type~1 is long and type~2 is short when re-trading is allowed. Notice that an agent who is short expects on average to cover some of her shorts in the future. When re-trading is ruled out, she prefers to cut her short position at time zero. This would place upward pressure on the static price. The long party also expects to reduce his position if $X(t)$ would stay constant, but because $b_1>0$, he expects the state $X(t)$ to grow, thus dampening his need to anticipate the reduction when re-trading is ruled out. In other words, the long party values the resale option less than the short party values the repurchase option. As a result, the static market presents excess demand at price $p_{\dyn}$ and thus the static price must rise to clear the market.

%%%%%%%%%%%%%%%%%%%%%%%%%%%%%%%%%%%%%%%%%%%%%%%
\section{A Planner with Limited Instruments}\label{se:planner}

In this section we show that our equilibrium can be explained through a planner's problem which sheds light on the PDE for the equilibrium in Theorem~\ref{th:PDE}. We first explain why this requires a planner with limited instruments. Consider a planner that can allocate the supply arbitrarily across types at any time and state. In addition, she can make arbitrary lump-sum numeraire transfers $\theta_i(T,\omega)$ to agents of type $i=1,\dots, n$ as well as a transfer $\theta_0(T,\omega)$ to the agents that are originally endowed with the supply provided these transfers add up to zero. Criterion~\eqref{eq:expectedNetPayoff} implies that the traders' utility functions are separable and linear in numeraire transfers. Hence the convexity of the cost function for holding assets guarantees that the supply is equally distributed across types in any Pareto optimum; i.e., $y_i(t, x)= \frac{s(t,x)}{n}$. This property of the asset allocation holds in the equilibrium of Theorem~\ref{th:PDE} when traders have homogeneous beliefs ($Q_{1}=\dots=Q_{n}$). In this case, it is clear that the equilibrium allocation is actually a Pareto optimum; the functional form of the utility function compensates for the lack of complete markets. However, this optimality does not hold when traders are heterogeneous and hold different asset positions in equilibrium: gambling using the asset has real costs and a social planner would like to rule them out. This general non-optimality of our equilibrium also holds if we use the ``belief neutral'' Pareto inefficiency criteria in \cite{BSX}.

Although our equilibrium is not Pareto optimal, we can characterize the equilibrium price as the optimal value for a planner with limited instruments and the equilibrium allocations as the associated allocations induced by this planner. Consider a planner that can use two instruments. The first is to assign ``total cost coefficients'' $\alpha_i(t,x)\in[\alpha_{-},\alpha_{+}]$ for each type~$i$ at each date and state~$(t,x)$. If agent~$i$ decides to go short $y$ units, she will be subsidized so that her effective cost is $c_{i}(t,x,y)=\frac 1 {2 \alpha_i(t,x)} y^2$, whereas if she goes long, she will be taxed to have the same effective cost. The second instrument is to give lump-sum numeraire subsidies or charge lump-sum taxes $\theta_i(T,\omega)$, $i=1,\dots,n$  to each type.   The planner must break even   so that any aggregate taxes collected must equal the net subsidies provided.

Given the assigned cost coefficients and lump-sum transfers, agents choose asset positions taking prices as given and the market settles on prices that equilibrate supply and demand. Since the objective function is separable in the numeraire, the optimal positions are independent of the lump-sum transfers: agent~$i$ maximizes the expected net payoff $E_{i}\big[ \int_{0}^{T} \Phi(t)\,dP(t) - \int_{0}^{T} c_{i}(t,X(t),\Phi(t))\,dt\big]$ from trading which is analogous to~\eqref{eq:expectedNetPayoff} except that the cost is now given by $c_{i}$.
%
%---
%
%
%
Recall that $I_{*}$ denotes the group of agents who go short in the equilibrium of Theorem~\ref{th:PDE}; cf.~\eqref{eq:optControlThm}.

\begin{theorem}\label{th:principal}
  (i) For any sufficiently regular\,\footnote{See Section~\ref{se:plannerProof} in the Appendix for further details. In particular, the assignment defined in~(ii) is sufficiently regular in this sense.} assignment $\alpha=(\alpha_{1},\dots,\alpha_{n})$ of the planner, there exists a unique equilibrium with a price function $v_{\alpha}\in C^{1,2}_{b}$. This function can be characterized by a linear PDE and by a Feynman--Kac representation; cf.~\eqref{eq:principalLinearPDE} and~\eqref{eq:principalFeynmanKac} in the Appendix.

  (ii) The planner can maximize the initial price by choosing $\alpha_{i}(t,x)=\alpha_{-}$ when $i\in I_{*}(t,x)$ and $\alpha_{i}(t,x)=\alpha_{+}$ when $i\in I_{*}^{c}(t,x)$. Under this assignment, the price and the asset allocation coincide with the equilibrium of Theorem~\ref{th:PDE}. Agents assigned $\alpha_+$ choose to go long and agents assigned $\alpha_-$ choose to go short, so that no taxes, subsidies or transfers are collected.
\end{theorem}

This theorem states that a planner facing the constraint $\alpha_i \in [\alpha_-, \alpha_+]$ on the total cost coefficients  and who wishes to maximize the initial price of the asset\footnote{Or, when $s>0$, maximize the welfare of the initial asset holders.} would assign  $\alpha_i = \alpha_-$ to the agents that in our original equilibrium choose to go short and $\alpha_+$ to the remaining agents. Given the assignment, equilibrium prices will be identical to the ones obtained in our original equilibrium.  This confirms the intuition from Proposition~\ref{pr:monotonicityCost} which states that the price is increased if costs for optimists (longs) are reduced and costs for pessimists (shorts) are increased: within the constraint, this allocation is the most favorable for the optimists and the least favorable for the pessimists. If the planner were not constrained to the interval $[\alpha_-, \alpha_+]$, she could typically attain an even higher initial price: she would put a tax on the pessimists and subsidize the optimists. In fact, Example~\ref{ex:plannerConstrainedEfficiency}  below shows that unless the planner is constrained to $[\alpha_-, \alpha_+]$, she can make all agents, including the initial asset holders, better off.

We show in the proof that the assertion of the theorem holds not only for the initial price but also for the price $v(t,x)$ at any time and state: the planner is \emph{time-consistent}; there is no need for a commitment device.

\begin{remark}\label{rk:plannerExplainsPDE}
  The planner's problem helps explain the PDE for the equilibrium in Theorem~\ref{th:PDE}. Indeed, \eqref{eq:mainPDE} can be seen as a supremum of linear PDEs parametrized by the groups~$I$. The linear PDE for a fixed group~$I$ is exactly the equation for the equilibrium price $v_{\alpha}$ resulting from the assignment given by $\alpha_{i}=\alpha_{-}$ when $i\in I$ and $\alpha^{*}_{i}=\alpha_{-}$ when $i\in I^{c}$. Thus, \eqref{eq:mainPDE} can be understood as an optimization over assignments of $\alpha_{-}$ and $\alpha_{+}$ to the different types.
\end{remark}

The following example shows that without the constraint $\alpha_i \in [\alpha_-, \alpha_+]$, the planner may be able to  improve the utility of all agents relative to the equilibrium utility. 
  
\begin{example}\label{ex:plannerConstrainedEfficiency}  
  Consider $n=2$ types in a market with constant supply $s=0$. The equilibrium portfolios satisfy $\phi_{1}=-\phi_{2}$ by market clearing, and we may assume that they are not identically equal to zero. The precise views and costs are not important for this example; for instance, we can take $b_{1}=-b_{2}>0$, $\sigma_{1}=\sigma_{2}>0$ and $\alpha:=\alpha_{+}=\alpha_{-}>0$.
  
  Consider a planner who can charge a tax on all types so that the agents face an effective cost coefficient $\tilde\alpha=\alpha/2$. It follows from Remark~\ref{rk:scaling} that the equilibrium price is unaffected by this symmetric scaling: $\tilde{v}=v$. In particular, the portfolios are related by $\tilde{\phi}_{i}= \tilde\alpha\cL^{i}\tilde v=(\alpha/2)\cL^{i}v=\phi_{i}/2$. As a consequence, the trading gains of type $i$, $X_{i}=\int_{0}^{T} \Phi_{i}(t)\,dP(t),$ become $\tilde X_{i}=X_{i}/2$ in the new equilibrium. Moreover, as 
  $\tilde{c}(\tilde\phi_{i})
  =\frac{1}{2\tilde\alpha} \tilde\phi_{i}^{2}
  =\frac12 \frac{1}{2\alpha}\phi_{i}^{2}=c(\phi_{i})/2$,
  the holding costs are also cut by half. As $\phi_{1}=-\phi_{2}$, have that $X_{1}=-X_{2}$ and $\tilde X_{1}=-\tilde X_{2}$; in particular, the difference $\Delta_{i}=X_{i}-\tilde X_{i}$ satisfies $\Delta_{1}=-\Delta_{2}$. As a result, the planner can transfer $\Delta_{i}$ to agent $i$ at a zero net cost.  
 After taking this transfer into account, the utility of both types is increased since the total gains are the same as before but the holding costs are cut by half. Moreover, the planner is left with the revenue from the taxes. She may, e.g., distribute the revenue to the agents as an additional lump sum.
\end{example}
 
 Using a continuity argument, a similar example can be constructed with a supply $s>0$. Then, the planner may distribute some of the revenue to the agents initially endowed with the supply to ensure that their utility is also increased.
 %NOTE: For the continuity argument one can pay a small deterministic part of the revenue (which is bounded away from zero under appropriate assumptions) to the traders, in addition to the random variable $\Delta_{i}$ calculated under $s=0$ (!). That will ensure that they do not lose utility. It may be possible to follow instead the above example and directly pay the $\Delta_{i}$ corresponding to $s>0$, but since the difference to the original $\Delta$ (from $s=0$) is unbounded, it is not clear if the planner can be even at the end that way.

%%%%%%%%%%%%%%%%%%%%%%%%%%
\section{Conclusion}\label{se:conclusion}

In this paper we considered a continuous-time model of trading among risk-neutral agents with heterogeneous beliefs. Agents face  quadratic costs-of-carry and as a consequence, their marginal valuation of the asset decreases when the magnitude of their position increases, as it would be the case for risk-averse agents. In previous models of heterogeneous beliefs, it was assumed that agents face a constant marginal cost-of-carry for a positive position and an infinite cost for a negative position. As a result, buyers benefit from a resale option and are willing to pay for an asset in excess of their own valuation of the dividends of that asset. Moreover, the supply does not affect the equilibrium price.  We show that when buyers face an increasing marginal cost-of-carry, in equilibrium, they may also value an option to delay. We illustrate with an example that even when shorting is impossible, this delay option may cause the market to equilibrate below the price that would prevail if agents were restricted to buy-and-hold strategies. Moreover, we introduce the possibility of short-selling and show how this gives pessimists the analogous options. 
In our model, the price depends on the supply.

We characterize the unique equilibrium of our model as the solution to a Hamilton--Jacobi--Bellman of a novel form and use this to derive several comparative statics: the price decreases with an increase in the supply of the asset, with an increase in the cost of carrying long positions, and with a decrease in the cost of carrying short positions. The conclusions of earlier models are shown to hold in a limiting case of the model when the quadratic cost-of-carry for long positions converges to zero. An example shows that a decrease in the cost of shorting and the consequent increase in the supply of shorts can deflate the bubble.

In our model the demand for the asset is satisfied by the sum of two components---the exogenous supply and the short positions of the market participants. While the shorts are determined endogenously,  supply is independent of the current price and the beliefs of agents. The data in \cite{CHW11} suggests that shorting played the overwhelming role in pricking the CDO bubble, but there are other episodes, such as the internet bubble, where investments in projects underlying the asset-class and sales by insiders played an important role in satisfying the demand by optimists.  For such episodes one would need to supplement the theory in this paper with an equilibrium model of supply.

Critics have indicted synthetic CDOs for the inordinate damage created by the subprime implosion, but it is not obvious what would have happened if synthetics had not been created. The spreads in the  ``safe''  tranches of cash CDOs would have been even more compressed. More ominously, the  numbers reported in \cite{CHW11} suggest that  generating the  amount of BBB Home Equity bonds  referenced in the synthetic CDOs would have required making an additional 2.5 \emph{trillion} dollars of subprime mortgage loans. This would have probably resulted in substantially more  new house construction and mortgage defaults. The model in this paper suggests that if a mechanism for shorting BBB HE bonds and CDO tranches had been created earlier, the subprime bubble would have been smaller.

%%%%%%%%%%%%%%%%%%%%%%%%%%%%%%%%%%%%%%%
\appendix

%%%%%%%%%%%%%%%%%%%%%%%%%%%%%%%%%%%%%%%

\section{Adding Linear Costs}\label{se:linearCosts}

In this section, we generalize the cost-of-carry by adding linear terms and discuss the corresponding changes in our main results. Broadly speaking, the generalized model does not alter the economic conclusions.

Indeed, let
\begin{equation}\label{eq:costOfCarryLin}
 c(y) = 
 \begin{cases}
   \frac{1}{2\alpha_{+}} y^{2} + \beta_{+}y,  & y\geq0,\\
   \frac{1}{2\alpha_{-}} y^{2} + \beta_{-}|y|,  & y<0,
 \end{cases}
\end{equation}
where $\beta_{-},\beta_{+}\geq0$ are constants; as discussed in the Introduction, the main case of interest is $\beta_{-}>0$ and $\beta_{+}=0$. While this cost function is still strictly convex, it fails to be differentiable at $y=0$ unless $\beta_{-}=\beta_{+}=0$.

Following the proof of Lemma~\ref{le:optPortfolio}, the optimal portfolio~\eqref{eq:optPortfolio} then becomes
\begin{equation}\label{eq:optPortfolioLin}
 \phi_{i}(t,x) = \begin{cases}
   \alpha_{+}(\cL^{i}v(t,x)-\beta_{+}), & \cL^{i}v(t,x)\geq \beta_{+},\\
   \alpha_{-}(\cL^{i}v(t,x)+\beta_{-}), & \cL^{i}v(t,x)\leq -\beta_{-},\\
   0,& \mbox{else}.
 \end{cases}  
\end{equation}
That is, there is an interval $[-\beta_{-},\beta_{+}]$ of values of $\cL^{i}v(t,x)$ where it is optimal to have a zero position, due to the kink in the function $c$.

The main PDE~\eqref{eq:mainPDE} needs to be adapted correspondingly. Indeed, instead of considering only the group $I$ of agents holding a short position, we now need to distinguish a group $J$ of agents holding a (strict) long position---the group $J$ may be smaller than the complement $I^{c}$. More precisely, the generalized PDE~\eqref{eq:mainPDE} reads as follows (the proof is analogous to Theorem~\ref{th:PDE}).

\begin{theorem}\label{th:PDELin}
 The unique equilibrium price function $v\in C^{1,2}_{b}$ can be characterized as the unique solution of the PDE
 \begin{equation}\label{eq:mainPDELin}
 \partial_{t}v(t,x) + \sup_{I\cap J = \emptyset} \Big(\mu_{I,J}(t,x)\partial_{x}v(t,x) + \tfrac12\tr\Sigma_{I,J}^{2}(t,x)\partial_{xx}v(t,x) - \kappa_{I,J}(t,x)\Big) = 0
 \end{equation}
 on $[0,T)\times\R^{d}$ with terminal condition $v(T,x)=f(x)$, where the supremum is taken over all disjoint subsets $I,J\subseteq \{1,\dots,n\}$ and the coefficients are defined as
 \begin{equation*}\label{eq:defMuLin}
     \mu_{I,J}(t,x) = \tfrac{\alpha_{-}}{|I|\alpha_{-} + |J|\alpha_{+}} \sum_{i\in I} b_{i}(t,x) + \tfrac{\alpha_{+}}{|I|\alpha_{-} + |J|\alpha_{+}}\sum_{i\in J}b_{i}(t,x),
 \end{equation*}
   \begin{equation*}\label{eq:defSigmaLin}
     \Sigma_{I,J}^{2}(t,x) = \tfrac{\alpha_{-}}{|I|\alpha_{-} + |J|\alpha_{+}} \sum_{i\in I} \sigma_{i}^{2}(t,x) + \tfrac{\alpha_{+}}{|I|\alpha_{-} + |J|\alpha_{+}}\sum_{i\in J} \sigma_{i}^{2}(t,x),
 \end{equation*}
 \begin{equation*}\label{eq:defCostLin} 
 \kappa_{I,J}(t,x) = \frac{s(t,x) - |I|\alpha_{-}\beta_{-} + |J|\alpha_{+}\beta_{+} }{|I|\alpha_{-} + |J|\alpha_{+}}.
 \end{equation*}  
\end{theorem}

In particular, the additional constants $\beta_{-},\beta_{+}$ enter only through the running cost~$\kappa_{I,J}$. It follows that the results on the comparative statics in Propositions~\ref{pr:monotonicitySupply} and~\ref{pr:monotonicityCost} remain valid; in addition, the equilibrium price function $v$ is increasing with respect to $\beta_{-}$ and decreasing with respect to~$\beta_{+}$.

In the limiting case of zero cost for long positions, we now need to send $\alpha_{+}\to \infty$ and $\beta_{+}\to 0$. Then, the result of Proposition~\ref{pr:limitLong} is unchanged; i.e., the limiting equilibrium price function is the solution of
$$
 \partial_{t}v + \sup_{i\in \{1,\dots,n\}} \left( b_{i}\partial_{x}v + \tfrac12\tr \sigma_{i}^{2}\partial_{xx}v\right) = 0.
$$
On the other hand, for the limit $\alpha_{-}\to 0$ of infinite cost for shorting, the result of Proposition~\ref{pr:limitShort} changes slightly because the long positions are subject to $\beta_{+}$ which becomes an additional running cost in the limiting equation
$$
 \partial_{t}v + \sup_{\emptyset\neq J\subseteq \{1,\dots,n\}} \left(\tfrac{1}{|J|}\sum_{i\in J} b_{i}\partial_{x}v + \tfrac12\tr \tfrac{1}{|J|}\sum_{i\in J} \sigma_{i}^{2}\partial_{xx}v - \frac{s}{|J|\alpha_{+}}-\beta_{+}\right) = 0.
$$
The results for the static equilibrium problem can be generalized with analogous changes.

%%%%%%%%%%%%%%%%%%%%%%%%%%%%%%%%%%%%%%%%%%%%%%%%%%%%%%%%%
\section{Heterogeneous Costs}\label{se:heterogeneousCosts}

In this section, we show how the equilibrium of Theorem~\ref{th:PDE} changes if the cost coefficients $\alpha_{-},\alpha_{+}$ depend on the type  rather than being the same for all agents. We write $\alpha^{i}_{-},\alpha^{i}_{+}$   for the coefficients of type~$i$. The following result shows that while the structure of the equilibrium remains similar, agents with lower costs have a larger influence on the coefficients of the PDE that determines the equilibrium price. 
\begin{theorem}\label{th:PDEhetCosts}
 The unique equilibrium price function $v\in C^{1,2}_{b}$ can be characterized as the unique solution of the PDE~\eqref{eq:mainPDE} with coefficients
\begin{equation*}\label{eq:defMuHet}
    \mu_{I}(t,x) = \frac{1}{\sum_{i\in I} \alpha^{i}_{-} + \sum_{i\in I^{c}} \alpha^{i}_{+}} \left( \sum_{i\in I} \alpha^{i}_{-}b_{i}(t,x) + \sum_{i\in I^{c}} \alpha^{i}_{+}b_{i}(t,x)\right),
\end{equation*}
  \begin{equation*}\label{eq:defSigmaHet}
    \Sigma^{2}_{I}(t,x) = \frac{1}{\sum_{i\in I} \alpha^{i}_{-} + \sum_{i\in I^{c}} \alpha^{i}_{+}} \left( \sum_{i\in I} \alpha^{i}_{-}\sigma^{2}_{i}(t,x) + \sum_{i\in I^{c}} \alpha^{i}_{+}\sigma^{2}_{i}(t,x)\right), \end{equation*}
\begin{equation*}\label{eq:defCostHet} 
\kappa_{I}(t,x) = \frac{s(t,x)}{\sum_{i\in I} \alpha^{i}_{-} + \sum_{i\in I^{c}} \alpha^{i}_{+}}.
\end{equation*}  
\end{theorem}

The proof is analogous to Theorem~\ref{th:PDE}. As in Lemma~\ref{le:optPortfolio}, the optimal portfolios are given by $\alpha^{i}_{\pm} \cL^{i}v(t,x)$. Thus, as expected, types with lower costs hold larger positions.

%%%%%%%%%%%%%%%%%%%%%%%%%%%%%%%%%%%%%%%%%%%%%%%%%%%%%%%%%%%%%%%%%%%%%

\section{Quadratic Costs on Values of Positions}\label{se:costsOnValue}

In this section, we briefly explain what changes if costs are quadratic in the monetary value of the portfolio rather than the size; i.e., the instantaneous cost-of carry is
$$
   c(P(t)\Phi(t)) \qquad\mbox{instead of}\qquad c(\Phi(t))
$$
where $c$ is quadratic as in~\eqref{eq:costOfCarry}. If the price $P(t)$ becomes zero, these costs become zero which leads to infinite demands by the agents and thus to non-existence of equilibria. Therefore, this discussion pertains to assets with a strictly positive price.

Similarly as in Lemma~\ref{le:optPortfolio} we can derive the first-order condition of optimality for the portfolio function $\phi_{i}$ of agent $i$ as
  \begin{equation*}%\label{eq:optPortfolio}
   \phi_{i}(t,x) = \frac{\alpha_{\sgn(\cL^{i}v(t,x))}}{v(t,x)^{2}}\cL^{i}v(t,x),
  \end{equation*}
a similar expression as in Lemma~\ref{le:optPortfolio} except for the additional division by~$v^{2}$. Using market clearing as in the proof of Theorem~\ref{th:PDE} then produces a PDE where the term $\kappa_{I}$ of~\eqref{eq:defCost}  receives an additional factor $v^{2}$. This new term can no longer be interpreted as a running cost and in general, the PDE cannot be written as an HJB equation similar to~\eqref{eq:mainPDE} because in such an equation the maximization is necessarily carried out over terms which  are linear in the $v$-variable. Thus, we do not expect to have an interpretation of equilibria through a stochastic control problem or a social planner.
A remarkable exception is $\kappa_{I}\equiv0$ which occurs in the case of zero net supply. Then, the PDE is exactly the same as~\eqref{eq:mainPDE} and hence the equilibrium price is also the same. The actual portfolios of the agents are not identical, but they only differ by the factor $v^{2}$. (To ensure a priori that equilibrium prices are positive, it suffices to assume that the payoff $f$ is positive and bounded away from zero. The comparison principle then shows that prices remain bounded away from zero at all times.)

%%%%%%%%%%%%%%%%%%%%%%%%%%%%%%%%%%%%%%%

\section{Optimal Control Representation}\label{se:controlRepresentation}

The PDE~\eqref{eq:mainPDE} is the Hamilton--Jacobi--Bellman equation of a stochastic optimal control problem where the controller can choose a subset $I\subseteq \{1,\dots,n\}$ at any time and state, and that choice determines the instantaneous drift and volatility coefficients $\mu_{I}$ and $\Sigma_{I}$ as well as the running cost $\kappa_{I}$.

To formulate this problem precisely, consider a filtered probability space carrying a $d'$-dimensional Brownian motion $W$ and let $\Theta$ be the collection of all (progressively measurable) processes $\cI$ with values in the family
of all subsets of $\{1,\dots,n\}$.\footnote{While this collection of control processes appears somewhat non-standard, there is no difficulty involved in defining it---this  family of subsets is simply a discrete set with $2^{n}$ elements; it can be identified with $\{0,1\}^{n}$.} For each $\cI\in\Theta$, let $X^{t,x}_{\cI}(r)$, $r\in[t,T]$ be the solution of the SDE
\begin{equation}\label{eq:controlledSDE}
  dX(r) = \mu_{\cI(r)}(r,X(r))\,dr + \Sigma_{\cI(r)}(r,X(r))\,dW(r), \quad X(t)=x
\end{equation}
on the time interval $[t,T]$. It follows from the assumptions on the coefficients $b_{i},\sigma_{i}$ that this SDE with random coefficients has a unique strong solution.\footnote{The coefficients $\mu_{\cI}$ and $\Sigma_{\cI}$ may be quite irregular as stochastic processes but the dependence with respect to the $x$-variable is Lipschitz continuous. See~\cite[Theorem~2.5.7, p.\,82]{Krylov.80} for a general result on existence and uniqueness under Lipschitz conditions.} Therefore, we may consider the control problem
\begin{equation}\label{eq:controlProblem}
  V(t,x) = \sup_{\cI\in\Theta} E\left[f(X^{t,x}_{\cI}(T)) - \int_{t}^{T} \kappa_{\cI(r)}(r,X^{t,x}_{\cI}(r))\,dr\right]
\end{equation}
for $(t,x)\in [0,T]\times\R^{d}$, which gives rise to a second characterization for the equilibrium price function $v$.
  
\begin{proposition}\label{pr:controlProblem}
  The equilibrium price function $v$ from Theorem~\ref{th:PDE} coincides with the value function $V$ of~\eqref{eq:controlProblem}. Moreover, an optimal control for~\eqref{eq:controlProblem} is given by $\cI_{*}(t)=I_{*}(t,X(t))$, where, as in~\eqref{eq:optControlThm},  
  \begin{equation}\label{eq:optControl}
    I_{*}(t,x) = \{i\in \{1,\dots,n\}:\, \cL^{i}v(t,x)<0\}.
  \end{equation}
\end{proposition}

\begin{proof}%[Proof of Proposition~\ref{pr:controlProblem}.]
  By Theorem~\ref{th:PDE}, the function $v\in C^{1,2}_{b}$ is a solution of the PDE~\eqref{eq:mainPDE} which is the Hamilton--Jacobi--Bellman equation of the control problem~\eqref{eq:controlProblem}. Moreover, $I_{*}(t,x)$ maximizes the Hamiltonian as noted after~\eqref{eq:equivPDE}. Thus, the verification theorem of stochastic control, cf.\ \cite[Theorem~IV.3.1, p.\,157]{FlemingSoner.06},  shows that $v$ is the value function and $\cI_{*}$ is an optimal control.
\end{proof}

%%%%%%%%%%%%%%%%%%%%%%%%%%%%%%%%%%%%%%
\section{Examples}\label{se:furtherEx}

In this section we discuss two examples in more detail. The calculations are carried out in Appendix~\ref{se:proofs}, together with the rest of the proofs.

The first example, already outlined in Example~\ref{ex:negativeBubble}, shows that the static price may exceed the dynamic price, even when short-selling is prohibited.

\begin{example}\label{ex:negativeBubbleAppendix}
  Consider $n=2$ types with volatility coefficients $\sigma_{i}=0$ and constant, opposing drifts
  $$
    b_{1}=1,\qquad b_{2}=-1.
  $$  
  The payoff function is $f(y)=y^{2}$ and the supply $s>0$ is constant. Moreover, $\alpha_{-}=0$ and $\alpha_{+}=1$.
  Then, the dynamic equilibrium price is
  $$
  p_{\dyn}=\begin{cases}
  x^{2}-sT/2, & |x|+ T/2 \leq s/4,\\
  (|x|+T)^{2}-sT, & |x|+ T/2 > s/4,
  \end{cases}
  $$
  and corresponding optimal portfolios in feedback form are given by
  $$
    \phi_{1}(t,x) = \begin{cases}
    0, & |x|+ (T-t)/2 > s/4, \; x<0,\\
    s/2+2x, & |x|+ (T-t)/2 \leq s/4,\\
    s, & |x|+ (T-t)/2 > s/4, \; x>0, 
    \end{cases}
  $$
  $$
    \phi_{2}(t,x) = \begin{cases}
    s, & |x|+ (T-t)/2 > s/4, \; x<0,\\
    s/2-2x, & |x|+ (T-t)/2 \leq s/4,\\
    0, & |x|+ (T-t)/2 > s/4, \; x>0.
    \end{cases}
  $$
  The static equilibrium price is
  \begin{align*}
    p_{\sta} =\begin{cases}
    x^{2}+ T^{2} -sT/2, & |x|\leq s/4,\\
    x^{2}+ T^{2} + 2|x|T-sT, & |x|> s/4,
  \end{cases}
  \end{align*}  
  and corresponding optimal portfolios are given by
  $$
    q_{1} = \begin{cases}
    0, & x<-s/4,\\
    s/2+2x, & |x|\leq s/4,\\
    s, & x>s/4, 
    \end{cases}
    \quad \quad
    q_{2} = \begin{cases}
    s, & x<-s/4,\\
    s/2-2x, & |x|\leq s/4,\\
    0, & x>s/4.
    \end{cases}
  $$
  The static equilibrium price exceeds the dynamic price; more precisely,
  $$
  p_{\sta}-p_{\dyn} = \begin{cases}
   T^{2}, & |x|\leq s/4-T/2,\\
   (s/2 - 2|x|)T, &  s/4-T/2 < |x| < s/4,\\
  0, & |x|\geq s/4.
  \end{cases}
  $$
\end{example}

Next, we discuss in more detail how the delay option effect explains the difference $p_{\sta}-p_{\dyn}$ in this example. To that end, it will be useful to record the portfolios as expected by the agents:
since $X(t)=x+b_{i}t$ $Q_{i}$-a.s.\ and $\Phi_{i}(t)=\phi_{i}(t,X(t))$, we have
	$$
	    Q_{1}\as,\quad \Phi_{1}(t) = \begin{cases}
	    0, & |x+t|+ (T-t)/2 > s/4, \; x+t<0,\\
	    s/2+2t+2x, & |x+t|+ (T-t)/2 \leq s/4,\\
	    s, & |x+t|+ (T-t)/2 > s/4, \; x+t>0, 
	    \end{cases}
	$$
	$$
	    Q_{2}\as,\quad \Phi_{2}(t) = \begin{cases}
	    s, & |x-t|+ (T-t)/2 > s/4, \; x-t<0,\\
	    s/2+2t-2x, & |x-t|+ (T-t)/2 \leq s/4,\\
	    0, & |x-t|+ (T-t)/2 > s/4, \; x-t>0.
	    \end{cases}
	$$
Below, we  abuse this notation and simply write $\Phi_{1}(t)$ and $\Phi_{2}(t)$ for the expressions on the right hand side. We consider various intervals for the initial value $x$; by symmetry, we may focus on $x\geq0$ without loss of generality. We also assume that $s>T$, mainly to avoid distinguishing even more cases.

\emph{Case 1: $x\geq s/4+T$.} In this regime, the expected dynamic portfolios $\Phi_{1}$ and $\Phi_{2}$ are constant, and thus the delay option is never exercised. The static portfolios coincide with their initial values, $q_{1}=s=\Phi_{1}(0)$ and $q_{2}=0=\Phi_{2}(0)$, and the static and dynamic prices are equal: $p_{\sta}=p_{\dyn}$.

\emph{Case 2: $s/4 \leq x< s/4+T$.} As before, $q_{1}=s=\Phi_{1}(0)$ and $\Phi_{1}$ is constant. However, $\Phi_{2}(t)$ equals zero initially but may become positive for $t$ close to~$T$ (for suitable parameter values). Nevertheless, type 2 does not choose to anticipate her trading in the static case,   because the cost-of-carry outweighs the expected gains---we still have $q_{2}=0=\Phi_{2}(0)$ and $p_{\sta}=p_{\dyn}$.

\emph{Case 3: $(s/4 - T/2)^{+} < x< s/4$.} Once again, $\Phi_{1}\equiv s$ is constant, $\Phi_{2}(0)=0$, and $\Phi_{2}$ increases for some $t>0$. Furthermore, the increase in type 2's position is larger for smaller $x.$ Type~2 now does anticipate some of that increase in the static case and for this reason $p_{\dyn}$ is now too low to be an equilibrium price. The increase in price changes the optimal portfolio for agents of type 1.   We are in the mixed case where portfolios and prices adjust. Type~1 decreases his initial position to $q_{1}=s/2+2x<s=\Phi_{1}(0)$ and type~2 increases her position to $q_{2}=s/2-2x>0=\Phi_{2}(0)$. At the same time, the static equilibrium price is augmented, $p_{\sta}-p_{\dyn}=(2x-s/2)T>0$. As $x$ decreases from $s/4$ to $s/4 - T/2$, this difference increases linearly from $0$ to $T^{2}$, and the portfolios $(q_{1},q_{2})$ change linearly from (s,0) to $(s-T,T)$. In summary, the elimination of the delay  option in the static case results in  portfolio adjustments and a price increase.

\emph{Case 4: $0 \leq x\leq s/4 - T/2$.} In this last regime, both $\Phi_{1}$ and $\Phi_{2}$ are increasing in time, so both types are exercising the delay option when re-trading is allowed and have an anticipatory motive when they can only trade at $t=0.$ The initial dynamic portfolios are $\Phi_{1}(0)=s/2+2x>0$ and $\Phi_{2}(0)=s/2-2x>0$. Since both types want to anticipate in the static case, the static price must be higher. More precisely, the aggregate excess demand at price $p_{\dyn}$ equals $2T^2$ and thus is independent of $x$. Since we are in the region where both types have positive demand, the marginal effect of an increase in price is $-1$, for each type. Thus, the price adjustment that is necessary to clear the static market is exactly $T^2$ for every value of $x$ in this region.

The next example illustrates the mechanics of the delay option when there is no cost-of-carry for long positions:  the most optimistic agent holds the entire market and the dynamic equilibrium price always exceeds the static one.

\begin{example}\label{ex:noCost}
Let $\alpha_{+}=\infty$ and $\alpha_{-}=1$. We consider $n=2$ types with drift coefficients 
$$
b_{1}=1,\qquad b_{2}=0
$$
and volatility coefficients $\sigma_{1}=\sigma_{2}=0$. The payoff is $f(y)=y^{2}$ and the initial value is $x=0$, so that the first type is more optimistic at any time. 

As in Proposition~\ref{pr:limitStatic}, the static equilibrium price is given by the optimist's expectation $e_{1}=E_{1}[f(X(T))]=T^{2}$. Following Proposition~\ref{pr:limitLong}, the same holds for the dynamic price, so that 
$
p_{\sta}=p_{\dyn}.
$
The static and dynamic portfolios of the pessimist are given by
$$
q_{2}=T^{-1}(e_{2}-p_{\sta})=- T,\qquad \phi_{2}(t,x)=\partial_{t}v(t,x)=- 2(x+T-t).
$$
Under $Q_{2}$, the state process $X\equiv 0$ is constant, so that $\Phi_{2}(t)=\phi_{2}(t,X(t))=-2(T-t)$ a.s. Thus, the static position $q_{2}=-T$ anticipates some of the increase from $\Phi_{2}(0)=-2T$ to $\Phi_{2}(T)=0$. However, this does not affect the static equilibrium price because an optimistic agent is indifferent to the size of her (nonnegative) position---the absence of a cost-of-carry for long positions allows the portfolios to adjust without affecting the prices.
\end{example}

%%%%%%%%%%%%%%%%%%%%%%%%%%%%%%%%%%%%%%%

\section{Proofs}\label{se:proofs}

This appendix collects the proofs for Sections~\ref{se:main}--\ref{se:planner} and Appendix~\ref{se:furtherEx}.

%%%%%%%%%%%%%%%%%%%%%%%%%%%%%%%%%%%%%%%
\subsection{Proofs for Section~\ref{se:main}}

Before proving the main result of Theorem~\ref{th:PDE}, we record two lemmas for later reference. The first one guarantees the passage from almost-sure to pointwise identities.

\begin{lemma}\label{le:fullSupport}
  For all $i\in\{1,\dots,n\}$ and all $t\in(0,T]$, the support of $X(t)$ under $Q_{i}$ is the full space $\R^{d}.$
\end{lemma}

\begin{proof}
  Recall that $X$ is the coordinate-mapping process on $\Omega= C([0,T],\R^{d})$. Since $b_{i}$ is bounded and $\sigma_{i}^{2}$ is uniformly parabolic, the support of $Q_{i}$ in is the set of all paths $\omega$ with $\omega(0)=x$; see \cite[Theorem~3.1]{StroockVaradhan.72}. The claim is a direct consequence.
\end{proof}

The second lemma provides an expression for the optimal portfolios.

\begin{lemma}\label{le:optPortfolio}
  Let $v\in C^{1,2}_{b}$ and consider the (price) process $P(t)=v(t,X(t))$. The portfolio defined by $\Phi_{i}(t)=\phi_{i}(t,X(t))$, where
  \begin{equation}\label{eq:optPortfolio}
   \phi_{i}(t,x) = \alpha_{\sgn(\cL^{i}v(t,x))}\cL^{i}v(t,x),
  \end{equation}
  is the unique\footnote{We recall that uniqueness is understood up to $(Q_{i}\times dt)$-nullsets.} optimal portfolio for type $i$.
\end{lemma}

\begin{proof}
  We note that $\Phi_{i}$ is admissible since $v\in C^{1,2}_{b}$. Let $\Phi$ be any admissible portfolio. By It\^o's formula,
  $$
    \int_{0}^{T} \Phi(t)\,dP(t) - \int_{0}^{T} c(\Phi(t))\,dt = \int_{0}^{T} \{\Phi(t) \cL^{i}v(t,X(t)) - c(\Phi(t))\}\,dt + M_{i}(T)
  $$
  where $M_{i}(T)$ is the terminal value of a (true) martingale with vanishing expectation; recall that $\sigma_{i}$ and $\partial_{x}v$ are bounded. Thus, the expected final payoff~\eqref{eq:expectedNetPayoff} is given by
  $$
    E_{i}\left[\int_{0}^{T} \{\Phi(t) \cL^{i}v(t,X(t)) - c(\Phi(t))\}\,dt\right].
  $$
  As a result, $\Phi$ is optimal if and only if it maximizes the above integrand (up to $(Q_{i}\times dt)$-nullsets). The unique maximizer is given by $\Phi_{i}$, and the claim follows.
\end{proof}

We can now prove the main result on the pricing PDE.

\begin{proof}[Proof of Theorem~\ref{th:PDE}.]
  (a) We first show that a given equilibrium price function $v\in C^{1,2}_{b}$ solves the PDE. Since $v(T,X(T))=f(X(T))$ $Q_{i}$-a.s.\ for all~$i$, the terminal condition $v(T,\cdot)=f$ follows from Lemma~\ref{le:fullSupport}. At any state $(t,x)$, we introduce the set
  \begin{equation}\label{eq:defIstar}
    I_{*}(t,x) = \{i\in\{1,\dots,n\}: \cL^{i}v(t,x)<0\}.
  \end{equation}
	Next, we recall the unique optimal portfolios $\Phi_{i}$ from Lemma~\ref{le:optPortfolio}. Using again Lemma~\ref{le:fullSupport}, the market clearing condition $\sum_{i}\Phi_{i}=S$ can be stated as
  \begin{equation}\label{eq:marketClearingForOptimal}
    \alpha_{-}\sum_{i\in I_{*}}\cL^{i}v + \alpha_{+}\sum_{i\in I_{*}^{c}}\cL^{i}v=s.
	\end{equation}
	If $i\in I_{*}$, then $\cL^{i}v\leq0$ and $\alpha_{-}\leq\alpha_{+}$ implies $\alpha_{-}\cL^{i}v\geq \alpha_{+} \cL^{i}v$. Conversely, if $i\in I_{*}^{c}$, then $\cL^{i}v\geq0$ and $\alpha_{+}\cL^{i}v\geq \alpha_{-} \cL^{i}v$. It follows that the set $I_{*}$ of~\eqref{eq:defIstar} maximizes the left hand side of~\eqref{eq:marketClearingForOptimal} among all subsets $I\subseteq \{1,\dots,n\}$. That is,
  \begin{equation}\label{eq:equivPDE}
    \max_{I\subseteq \{1,\dots,n\}} \left(\alpha_{-}\sum_{i\in I}\cL^{i}v + \alpha_{+}\sum_{i\in I^{c}}\cL^{i}v -s \right)=0
  \end{equation}
	and the set $I_{*}$ is a maximizer, or equivalently,
  \begin{equation}\label{eq:equivPDE2}
    \max_{I\subseteq \{1,\dots,n\}} \tfrac{1}{|I|\alpha_{-} + |I^{c}|\alpha_{+}} \left(\alpha_{-}\sum_{i\in I}\cL^{i}v + \alpha_{+}\sum_{i\in I^{c}}\cL^{i}v -s \right)=0.
  \end{equation}
	After plugging in the definition of $\cL^{i}v$ and using the definitions of $\mu_{I}$, $\Sigma_{I}$ and $\kappa_{I}$ in \eqref{eq:defMu}--\eqref{eq:defCost}, this is the desired PDE~\eqref{eq:mainPDE}.
	
	\vspace{.5em}
	
	(b) Conversely, let $v\in C^{1,2}_{b}$ be a solution of the PDE~\eqref{eq:mainPDE} with terminal condition $f$ and define $\Phi_{i},\phi_{i}$ as in part~(i) of Theorem~\ref{th:PDE}. Then, the terminal condition $v(T,X(T))=f(X(T))$ is satisfied and $\Phi_{i}$ are optimal by Lemma~\ref{le:optPortfolio}. Since $v$ is a solution of the equivalent PDE~\eqref{eq:equivPDE} and $I_{*}$ of~\eqref{eq:defIstar} is a maximizer, we have that 
  $$
    \sum_{1\leq i \leq n} \phi_{i} = \alpha_{-}\sum_{i\in I_{*}}\cL^{i}v + \alpha_{+}\sum_{i\in I_{*}^{c}}\cL^{i}v = s;
  $$
  that is, the market clears. This shows that $v$ is an equilibrium price function.

	\vspace{.5em}
	  
  (c) Since (a) and (b) established a one-to-one correspondence between equilibria and solutions of the PDE~\eqref{eq:mainPDE} with terminal condition $f$, it remains to observe that the latter has a unique solution in $C^{1,2}_{b}$. Indeed, existence holds by\footnote{The gist of this rather technical result is that a second-order parabolic PDE of HJB-type has a solution in $C^{1,2}_{b}$ as soon as the second-order term is uniformly parabolic and all coefficients and boundary conditions are sufficiently smooth and bounded.} \cite[Theorem~6.4.3, p.\,301]{Krylov.87}; the conditions in the cited theorem can be verified as in~\cite[Example~6.1.4, p.\,279]{Krylov.87}.
  
  Uniqueness holds by the comparison principle for parabolic PDEs; see \cite[Theorem~V.9.1, p.\,223]{FlemingSoner.06}.
\end{proof}

%%%%%%%%%%%%%%%%%%%%%%%%%%%%%%%%%%%%%%%
\subsection{Proofs for Section~\ref{se:comparativeStatics}}

We start with the comparative statics for the dependence of the price on the supply.

\begin{proof}[Proof of Proposition~\ref{pr:monotonicitySupply}.]
  Since the function $s$ enters linearly in the running cost~\eqref{eq:defCost} of the control problem~\eqref{eq:controlProblem} and nowhere else, it follows immediately that the value function $V$ is monotone decreasing in $s$, and now the claim follows from Proposition~\ref{pr:controlProblem}.
\end{proof}

Next, we consider the dependence on the cost coefficients.

\begin{proof}[Proof of Proposition~\ref{pr:monotonicityCost} and Remark~\ref{rk:preciseComparison}.]
  Let $\alpha_{-}\leq\alpha_{+}$ and $\alpha'_{-}\leq\alpha'_{+}$ be two pairs of cost coefficients and let $v$ and $v'$ be the corresponding equilibrium price functions. Let $I_{*}$ be the optimal feedback control for $\alpha_{\pm}$ as defined in~\eqref{eq:defIstar}, then by~\eqref{eq:equivPDE} we have
  $$
      \alpha_{-}\sum_{i\in I_{*}}\cL^{i}v + \alpha_{+}\sum_{i\in I_{*}^{c}}\cL^{i}v -s =0.
  $$
  If $\alpha'_{-}\leq\alpha_{-}$ and $\alpha'_{+}\geq\alpha_{+}$, then $\sum_{i\in I_{*}}\cL^{i}v\leq0$ and $\sum_{i\in I_{*}^{c}}\cL^{i}v \geq0$ yield that
  $$
      \alpha'_{-}\sum_{i\in I_{*}}\cL^{i}v + \alpha'_{+}\sum_{i\in I_{*}^{c}}\cL^{i}v -s \geq0.
  $$
  In the special case where $s\equiv0$, this conclusion also holds under the weaker condition that $\alpha_{+}/\alpha_{-}\leq \alpha'_{+}/\alpha'_{-}$, which covers the case~(iii), and the same holds under the conditions of Remark~\ref{rk:preciseComparison}. It then follows that
  $$
    \max_{I\subseteq \{1,\dots,n\}} \left(\alpha'_{-}\sum_{i\in I}\cL^{i}v + \alpha'_{+}\sum_{i\in I^{c}}\cL^{i}v -s \right)\geq0,
  $$
  which is a version of~\eqref{eq:equivPDE} with inequality instead of equality, for the coefficients $\alpha'_{\pm}$. Following the same steps as after~\eqref{eq:equivPDE}, we deduce that 
  $$
    \partial_{t}v + \sup_{I\subseteq \{1,\dots,n\}} \Big(\mu'_{I}\partial_{x}v + \tfrac12 \tr\Sigma_{I}^{'2}\partial_{xx}v - \kappa'_{I} \Big) \geq0,
  $$
  where $\mu'_{I}, \Sigma'_{I}, \kappa'_{I}$ are defined as in \eqref{eq:defMu}--\eqref{eq:defCost} but with $\alpha'_{\pm}$ instead of $\alpha_{\pm}$. In other words, $v$ is a subsolution\footnote{Note that the sign convention chosen here is opposite to the one of~\cite{FlemingSoner.06}, so that a subsolution corresponds to the inequality $\geq0$ in the PDE.} of the PDE satisfied by $v'$. As $v$ and $v'$ satisfy the same terminal condition $f$, the comparison principle \cite[Theorem~V.9.1, p.\,223]{FlemingSoner.06} implies that $v\leq v'$.
\end{proof}

We continue with our result on the limit $\alpha_{+}\to\infty$.

\begin{proof}[Proof of Proposition~\ref{pr:limitLong}.]
  We first notice that since $s\geq0$, the optimal set $I_{*}$ of~\eqref{eq:defIstar} for the Hamiltonian of the PDE~\eqref{eq:mainPDE} must satisfy $|I_{*}|<n$ due to the market clearing condition---at least one agent has to hold a nonnegative position. As a result, the PDE~\eqref{eq:mainPDE} remains the same if the supremum is taken only over sets $I$ with $|I^{c}|>0$.
  
  Taking that into account, the limiting PDE for~\eqref{eq:mainPDE} as $\alpha_{+}\to \infty$ is
  \begin{equation}\label{eq:PDElimitProof}
  \partial_{t}v + \sup_{\emptyset\neq J\subseteq \{1,\dots,n\}}  \tfrac{1}{|J|}\sum_{i\in J} \left(b_{i}\partial_{x}v + \tfrac12\tr \sigma_{i}^{2}\partial_{xx}v \right) = 0.
  \end{equation}  
  Notice that given a set of real numbers, the largest average over a subset is in fact equal to the largest number in the set. As a result, \eqref{eq:PDElimitProof} coincides with~\eqref{eq:PDEalphaPlusInfty}. Using again \cite[Theorem~6.4.3, p.\,301]{Krylov.87} and \cite[Theorem~V.9.1, p.\,223]{FlemingSoner.06}, this equation has a unique solution $v^{\infty}\in C^{1,2}_{b}$ for the terminal condition $f$, and the solution is independent of $\alpha_{-}$  and $s$ since these quantities do not appear in~\eqref{eq:PDEalphaPlusInfty}. 
  
  To see that $v^{\alpha_{-},\alpha_{+}}(t,x)\to v^{\infty}(t,x)$, one can apply a PDE technique called the Barles--Perthame procedure to the equations under consideration; see~\cite[Section~VII.3]{FlemingSoner.06}. Alternately, and to give a more concise proof, we may use the representation~\eqref{eq:controlProblem} of $v^{\alpha_{-},\alpha_{+}}$ as a value function as well as the corresponding representation for $v^{\infty}$. Then, a result on the stability of value functions, cf.\ \cite[Corollary~3.1.13, p.\,138]{Krylov.80}, shows that $v^{\alpha_{-},\alpha_{+}}\to v^{\infty}$ locally uniformly; that is,
 $$
   \sup_{(t,x)\in[0,T]\times K} |v^{\alpha_{-},\alpha_{+}}(t,x) - v^{\infty}(t,x)| \to 0
 $$
 for any compact set $K\subseteq \R^{d}$. The monotonicity property of the limit follows from Proposition~\ref{pr:monotonicityCost}.
\end{proof}

Finally, we turn to the limit $\alpha_{-}\to0$.

\begin{proof}[Proof of Proposition~\ref{pr:limitShort}.]
  The arguments are similar to the ones for Proposition~\ref{pr:limitLong}. In this case, the limiting PDE for~\eqref{eq:mainPDE} as $\alpha_{-}\to 0$ is~\eqref{eq:PDEalphaMinusZero}. As in the proof of Proposition~\ref{pr:limitLong} we have that the limiting PDE has a unique solution $v^{0,\alpha_{+}}\in C^{1,2}_{b}$ and $v^{\alpha_{-},\alpha_{+}}(t,x)\to v^{0,\alpha_{+}}(t,x)$ locally uniformly, with monotonicity in $\alpha_{-}$. In the special case where $s=0$, the PDE~\eqref{eq:PDEalphaMinusZero} coincides with~\eqref{eq:PDElimitProof}, and thus with~\eqref{eq:PDEalphaPlusInfty} as shown in the proof of Proposition~\ref{pr:limitLong}.
\end{proof}

%%%%%%%%%%%%%%%%%%%%%%%%%%%%%%%%%%%%%%%%%%%%%%%%%%%%
\subsection{Proofs for Section~\ref{se:speculation} and Appendix~\ref{se:furtherEx}}

We first prove our formula for the static equilibrium price.

\begin{proof}[Proof of Proposition~\ref{pr:staticPrice}.]
  We set $e_{i} = E_{i}[f(X(T))]$. Given any price $p$, the expected net payoff for agent $i$ using portfolio $q$ is
$$
  q (e_{i}-p) - \tfrac{T}{2\alpha_{\sgn(q)}} q^{2}
$$
and the unique maximizer is $q_{i} = \alpha_{\sgn(e_{i}-p)} T^{-1} (e_{i}-p)$ as stated in~\eqref{eq:staticPortfolio}. 

  Let $p$ be a static equilibrium price. Setting $I_{*}=\{i\in\{1,\dots,n\}:\, e_{i}<p\}$, the market clearing condition $\sum_{i} q_{i}=s$ for these optimal portfolios yields
	\begin{equation}\label{eq:staticPriceClears}
	  \alpha_{-} \sum_{i\in I_{*}} (e_{i}-p) + \alpha_{+} \sum_{i\in I_{*}^{c}} (e_{i}-p) =sT
  \end{equation}
	and we observe that $I_{*}$ maximizes the left hand side; that is,
	$$
	  \max_{I\subseteq \{1,\dots,n\}} \left(\alpha_{-} \sum_{i\in I} (e_{i}-p) + \alpha_{+} \sum_{i\in I^{c}} (e_{i}-p) - sT\right)=0.
  $$
  This is equivalent to the claimed representation~\eqref{eq:staticPrice} for $p$.
  
  Conversely, define $p$ by~\eqref{eq:staticPrice} and $q_{i}$ by~\eqref{eq:staticPortfolio}, then $q_{i}$ is optimal for agent~$i$ as mentioned in the beginning of the proof. Moreover, reversing the above, $p$ satisfies~\eqref{eq:staticPriceClears} and thus 
  $$
    \sum_{i=1}^{n} q_{i} = \alpha_{-} \sum_{i\in I_{*}} T^{-1}(e_{i}-p) + \alpha_{+} \sum_{i\in I_{*}^{c}} T^{-1}(e_{i}-p) =s,
  $$
  establishing market clearing. 
\end{proof}

We can now deduce the formulas for the limiting cases of the static equilibrium.

\begin{proof}[Proof of Proposition~\ref{pr:limitStatic}.]
  Formula~\eqref{eq:StaticPriceAlphaPlusInfty} follows by taking  the limit $\alpha_{+}\to \infty$  in~\eqref{eq:staticPrice}. Similarly, \eqref{eq:StaticPriceAlphaMinusZero} is obtained by taking the limit $\alpha_{-}\to 0$ in~\eqref{eq:staticPrice}. 
\end{proof}

Next, we show that in the limit $\alpha_{+}\to \infty$ with no cost on long positions, the dynamic price exceeds the static one.

\begin{proof}[Proof of Proposition~\ref{pr:limitLongComparison}.]
  By the formula~\eqref{eq:StaticPriceAlphaPlusInfty}
 for $p_{\sta}^{\infty}$, it suffices to verify that $E_{i}[f(X(T))]\leq p_{\dyn}^{\infty}$ for fixed $i\in\{1,\dots,n\}$.
  Indeed, let $u=u_{i}\in C^{1,2}_{b}$ be the unique solution of
  $$
  \partial_{t}u + b_{i}\partial_{x}u + \tfrac12\tr \sigma_{i}^{2}\partial_{xx}u = 0,\quad u(T,\cdot)=f.
  $$
  Then, by the Feynman--Kac formula \cite[Theorem~5.7.6, p.\,366]{KaratzasShreve.91}, we have $u(0,x)=E_{i}[f(X(T))]$. Moreover, $u$ is clearly a subsolution of the PDE~\eqref{eq:PDEalphaPlusInfty} for $v^{\infty}$, and now the comparison principle \cite[Theorem~V.9.1, p.\,223]{FlemingSoner.06} yields that 
  $E_{i}[f(X(T))]=u(0,x)\leq v^{\infty}(0,x)=p_{\dyn}^{\infty}$ as claimed.
\end{proof}

In what follows, we show that in the limit $\alpha_{-}\to 0$ where short-selling is prohibited, the same inequality holds, provided one agents holds the static market.

\begin{proof}[Proof of Proposition~\ref{pr:comparisonOneOptimist}.]
  In view of~\eqref{eq:StaticPriceAlphaMinusZero}, we have $p_{\sta}^{0,\alpha_{+}}=E_{i}[f(X(T))]
  - \frac{sT}{\alpha_{+}}$ since the maximizing set is $J=\{i\}$. Using again the Feynman--Kac formula \cite[Theorem~5.7.6, p.\,366]{KaratzasShreve.91}, we deduce that $p_{\sta}^{0,\alpha_{+}}=u(0,x)$ where $u\in C^{1,2}_{b}$ is the solution of
  $$
  \partial_{t}u + b_{i}\partial_{x}u + \tfrac12\tr \sigma_{i}^{2}\partial_{xx}u - \frac{sT}{\alpha_{+}} = 0,\quad u(T,\cdot)=f.
  $$
  In particular, $u$ is a subsolution of the PDE~\eqref{eq:PDEalphaMinusZero} for $v^{0,\alpha_{+}}$, and now the comparison principle \cite[Theorem~V.9.1, p.\,223]{FlemingSoner.06} yields that 
  $p_{\sta}^{0,\alpha_{+}}=u(0,x)\leq v^{0,\alpha_{+}}(0,x)=p_{\dyn}^{0,\alpha_{+}}$ as desired.
\end{proof}

We turn to our example where the static price exceeds the dynamic one due to the delay option effect.

\begin{proof}[Proofs for Example~\ref{ex:negativeBubble} (Example~\ref{ex:negativeBubbleAppendix}).]
  We begin with the static case. For later use, we consider the more general situation where $\sigma:=\sigma_{1}=\sigma_{2}$ may be positive (but constant). We have
  $
  e_{i}= E_{i}[f(X(T))] = x^{2} + 2x b_{i}T + T^{2} + \sigma^{2}T
	$
	and thus, as in~\eqref{eq:StaticPriceAlphaMinusZero}, the static price $p_{\sta}$ is 
	\begin{align*}
	    p_{\sta} &= \max_{\emptyset\neq J \subseteq \{1,2\}} \left(\tfrac{1}{|J|} \sum_{i\in J} e_{i} - \frac{sT}{|J|}\right)\\
	  &= x^{2}  + \sigma^{2}T + \max\left\{T^{2} -sT/2, T^{2}+2|x|T-sT\right\}
	  \end{align*}
	  or
	  \begin{align}\label{eq:exampleStaticPriceProof}
	   p_{\sta}=\begin{cases}
	    x^{2}+ \sigma^{2}T + T^{2} -sT/2, & |x|\leq s/4,\\
	    x^{2}+ \sigma^{2}T + T^{2} + 2|x|T-sT, & |x|> s/4
	  \end{cases}
	\end{align}
	and the portfolios $q_{i}$ are as stated in Example~\ref{ex:negativeBubbleAppendix}.
  
  We turn to the dynamic case and restrict to $\sigma=0$. The limiting equation for~\eqref{eq:PDEalphaMinusZero} is
	\begin{equation}\label{eq:zeroVolPDE}
	  \partial_{t}v + \max \left(|\partial_{x}v|- s, -s/2\right) = 0,\quad v(T,\cdot)=f.
	\end{equation}
	In analogy to Proposition~\ref{pr:controlProblem}, this can be seen as the Hamilton--Jacobi equation of a deterministic control problem where the drift $\mu$ of the controlled state $dX(t)=\mu(t,X(t))\,dt$ can be chosen to be $\pm 1$ or $0$ and the running cost is $s$ or $s/2$, respectively. We can check directly that an optimal control for this problem is
	$$
	  \mu(t,x)=\begin{cases}
	  \sgn(x), &|x|+ (T-t)/2 > s/4,\\
	  0, & |x|+ (T-t)/2 \leq s/4,
	  \end{cases}
	$$
	and then the value function is found to be
	$$
	  v(t,x)=\begin{cases}
	  (|x|+T-t)^{2}-s(T-t), & |x|+ (T-t)/2 > s/4,\\
	  x^{2}-s(T-t)/2, & |x|+ (T-t)/2 \leq s/4.
	  \end{cases}
	$$
	Indeed, $v$ is continuous and the unique viscosity\footnote{
	As is often the case for deterministic control problems, the value function is not $C^{1,1}$ and~\eqref{eq:zeroVolPDE} has no classical solution.} % 
	solution of~\eqref{eq:zeroVolPDE}. The indicated formulas for $p_{\dyn}-p_{\sta}=v(0,x)-p_{\sta}$ and for the optimal controls $\phi_{i}$ follow. 
\end{proof}

Next, we prove the continuity of the prices in the small volatility limit.

\begin{proof}[Proof of Proposition~\ref{pr:limitSmallVolatilityEx}.]
  For the static case, the formula for $p_{\sta}^{\sigma}$ stated in~\eqref{eq:exampleStaticPriceProof} shows that $p_{\sta}^{\sigma} - p_{\sta}^{0}=\sigma^{2}T\downarrow0$. Turning to the dynamic case, we first show that $p_{\dyn}^{\sigma}=v^{\sigma}(0,x)$ is monotone with respect to $\sigma$. Since $f$ is convex, $x\mapsto v^{\sigma}(t,x)$ is convex and thus $\partial_{xx}v^{\sigma}\geq0$. Given $\sigma_{1}\geq\sigma_{2}>0$, it follows that $v^{\sigma_{2}}$ is a subsolution to the equation~\eqref{eq:PDEalphaMinusZero} for $v^{\sigma_{1}}$, and thus the comparison principle for parabolic PDEs implies that $v^{\sigma_{1}}\geq v^{\sigma_{2}}$. To see that $v^{\sigma}(t,x)\to v^{0}(t,x)$ and in particular $p_{\dyn}^{\sigma}\to p_{\dyn}^{0}$, we may again use a general result on the stability of value functions; cf.\ \cite[Corollary~3.1.13, p.\,138]{Krylov.80}.
\end{proof}

It remains to provide the calculations for our symmetric example with $\alpha_{-}=\alpha_{+}=1$.

\begin{proof}[Proofs for Example~\ref{ex:symmetricCosts}.]
  Following Remark~\ref{rk:symmetricCost}, the equilibrium price function in the dynamic case is
  $$
    v(t,x)=E[f(x+ \mu\tau+ \Sigma B_{\tau})], \quad\mbox{where}\quad \tau:=T-t
  $$
  and $B_{\tau}$ is a centered Gaussian with variance $\tau$.
  As $f(y)=y^{2}$,
  $$
    v(t,x) = x^{2} + 2x\mu\tau + \mu^{2}\tau^{2} + \Sigma^{2}\tau
  $$
  and the optimal portfolios in feedback form are given by
  $$
    \phi_{i}(t,x)=\cL^{i}v(t,x) = x(b_{i}-b_{j}) + \tfrac12 \tau (b_{i}^{2}-b_{j}^{2}) + \tfrac12 (\sigma_{i}^{2}-\sigma_{j}^{2}).
    $$
  For the static case, we have
  $$
    e_{i}= x^{2} + 2x b_{i}T + b_{i}^{2}T^{2} + \sigma_{i}^{2}T
  $$
  and thus 
  $$
    p_{\sta} = \frac{e_{1}+e_{2}}{2} = x^{2} + 2x \mu T + \frac{b_{1}^{2}+b_{2}^{2}}{2} T^{2} + \Sigma^{2}T
  $$
  as well as
  $$
    q_{i} = T^{-1}(e_{i}-p_{\sta}) = T^{-1}\frac{e_{i}-e_{j}}{2} = x(b_{i}-b_{j}) + \tfrac12 T(b_{i}^{2}-b_{j}^{2}) + \tfrac12 (\sigma_{i}^{2}-\sigma_{j}^{2}).
  $$
\end{proof}

%%%%%%%%%%%%%%%%%%%%%%%%%%%%%%%%%%%%%%%
\subsection{Proofs for Section~\ref{se:planner}}\label{se:plannerProof}

In this section we discuss the planner's problem.

\begin{proof}[Proof of Theorem~\ref{th:principal}.]
  Let $\cA$ be the set of all assignments; that is, all measurable functions $\alpha(t,x)=(\alpha_{1}(t,x),\dots,\alpha_{n}(t,x))$ with $\alpha_{-}\leq \alpha_{i}\leq \alpha_{+}$.
  
  (i)  Fix $\alpha\in\cA$ and denote $\|\alpha\|=\sum_{i=1}^{n}\alpha_{i}$. 
  Suppose first that $w\in C^{1,2}_{b}$ is a given equilibrium price function for $\alpha$. As in Lemma~\ref{le:optPortfolio}, the unique optimal portfolio for agent~$i$ is $\Phi_{i}(t)=\phi_{i}(t,X(t))$ where
  \begin{equation}\label{eq:optPortfolioPrincipal}
   \phi_{i}(t,x) = \alpha_{i}(t,x)\cL^{i}w(t,x).
  \end{equation}
  The market clearing condition then implies 
  $ 
    \sum_{i=1}^{n} \alpha_{i} \cL^{i}w =s
	$
	which is equivalent to
  \begin{equation}\label{eq:principalPDElinearAlpha}
    \frac{1}{\|\alpha\|}\left(\sum_{i=1}^{n} \alpha_{i} \cL^{i}w -s \right) =0
  \end{equation}
	or
  \begin{equation}\label{eq:principalLinearPDE}
  \partial_{t}w + \frac{1}{\|\alpha\|}\sum_{i=1}^{n} \alpha_{i}b_{i} \partial_{x}w + \frac12\tr\frac{1}{\|\alpha\|}\sum_{i=1}^{n}\alpha_{i}\sigma^{2}_{i} \partial_{xx}w - \frac{s}{\|\alpha\|} = 0.
\end{equation}
  Together with the terminal condition $w(T,\cdot)=f$, this implies by It\^o's formula that~$w$ has the Feynman--Kac representation
  \begin{equation}\label{eq:principalFeynmanKac}
    w(t,x)=E\left[f(X^{t,x}_{\alpha}(T)) - \int_{t}^{T} \kappa_{\alpha}(r,X^{t,x}_{\alpha}(r))\,dr\right]
\end{equation}
where $\kappa_{\alpha}=s/\|\alpha\|$ and $X^{t,x}_{\alpha}$ is a diffusion with initial condition $X^{t,x}(t)=x$, drift  $\mu_{\alpha}=\frac{1}{\|\alpha\|}\sum_{i=1}^{n} \alpha_{i}b_{i}$ and volatility  $\sigma_{\alpha}=\frac{1}{\|\alpha\|}\sum_{i=1}^{n}\alpha_{i}\sigma_{i}$. %In particular, $w$ is the unique solution of~\eqref{eq:principalLinearPDE} in $C^{1,2}_{b}$. 

Conversely, suppose that $\alpha$ is sufficiently regular so that~\eqref{eq:principalLinearPDE} has a solution $w\in C^{1,2}_{b}$, then reversing the above arguments shows that $w$ is an equilibrium price function given the assignment~$\alpha$. For examples of sufficient regularity conditions on $\alpha$ see e.g.\ \cite[p.\,147]{Friedman.75}.

%
%  If $\alpha$ is sufficiently regular (e.g., uniformly H\"older continuous), then given the terminal condition $f$, this linear PDE has a unique solution in $C^{1,2}_{b}$ which by the Feynman--Kac formula \cite[Theorem~5.7.6, p.\,366 and subsequent remark]{KaratzasShreve.91} has the representation
%  \begin{equation}\label{eq:principalFeynmanKac}
%    w(t,x)=E\left[f(X^{t,x}_{\alpha}(T)) - \int_{t}^{T} \kappa_{\alpha}(r,X^{t,x}_{\alpha}(r))\,dr\right]
%\end{equation}
%where $\kappa_{\alpha}=s/\|\alpha\|$ and $X^{t,x}_{\alpha}$ is the diffusion with initial condition $X^{t,x}(t)=x$, drift  $\mu_{\alpha}=\frac{1}{\|\alpha\|}\sum_{i=1}^{n} \alpha_{i}b_{i}$ and volatility  $\sigma_{\alpha}=\frac{1}{\|\alpha\|}\sum_{i=1}^{n}\alpha_{i}\sigma_{i}$.
%
%  Conversely, starting with the unique solution $w\in C^{1,2}_{b}$ of~\eqref{eq:principalLinearPDE}, reversing the above arguments shows that $w$ is an equilibrium price function given the assignment~$\alpha$, and that completes the proof of (i) (when $\alpha$ is sufficiently regular (i.e., so that~\eqref{eq:principalLinearPDE} has a solution in $C^{1,2}_{b}$).  
%  %(ii) %In view of~\eqref{eq:principalFeynmanKac}, \eqref{eq:controlProblem} and Proposition~\ref{pr:controlProblem}, the supremum of $v_{\cI}(t,x)$ over all assignments $\cI(t)=I(t,X(t))$ is given by $V(t,x)$ and~\eqref{eq:optControl} is an optimal assignment in feedback form.
  
  (ii) Consider the nonlinear PDE
  $$
    \sup_{\alpha\in\cA} \sum_{i=1}^{n} \alpha_{i} \cL^{i}w  = s.
  $$
  Noting that the supremum is attained for $\alpha_{i}=\alpha_{\sgn(\cL^{i}w)}$, we observe that this is the same PDE as~\eqref{eq:marketClearingForOptimal}. To wit, after stating it in the equivalent form
  \begin{equation}\label{eq:plannerPDE}
    \sup_{\alpha\in\cA} \frac{1}{\|\alpha\|} \left(\sum_{i=1}^{n} \alpha_{i} \cL^{i}w - s \right)=0
  \end{equation} 
  we see that this is just a rewriting of~\eqref{eq:mainPDE}. In particular, for the terminal condition $f$, the unique solution of~\eqref{eq:plannerPDE} in $C^{1,2}_{b}$ is given by the equilibrium price function~$v$ of Theorem~\ref{th:PDE} and the stated assignment associated with $I_{*}$ attains that price. (As $v\in C^{1,2}_{b}$, this assignment is indeed ``sufficiently regular'' in the sense used in~(i) above.)
  
  To see that any other assignment leads to a lower price, consider a fixed (and sufficiently regular) assignment $\alpha=(\alpha_{1},\dots,\alpha_{n})$ and its equilibrium price~$v_{\alpha}$. As $v_{\alpha}$ solves~\eqref{eq:principalPDElinearAlpha},
%  $$
%    \sup_{\alpha'\in\cA} \sum_{i=1}^{n} \alpha'_{i} \cL^{i}v_{\alpha} - s \geq \sum_{i=1}^{n} \alpha_{i} \cL^{i}v_{\alpha} - s=0
%  $$
%  and thus
  $$
    \sup_{\alpha'\in\cA} \frac{1}{\|\alpha'\|} \left(\sum_{i=1}^{n} \alpha'_{i} \cL^{i}v_{\alpha} - s \right) \geq \frac{1}{\|\alpha\|} \left(\sum_{i=1}^{n} \alpha_{i} \cL^{i}v_{\alpha} - s \right) =0.
  $$
  This shows that $v_{\alpha}$ is a subsolution of~\eqref{eq:mainPDE} and hence $v_{\alpha}\leq v$ by the comparison principle of \cite[Theorem~V.9.1, p.\,223]{FlemingSoner.06}.
\end{proof}

\begin{remark}\label{rk:irregularAssignment}
  The difference for a general (measurable)~$\alpha$ is that the smoothness of the solution to the PDE~\eqref{eq:principalLinearPDE} is not clear. However, one may substitute the classical solution of the PDE by a suitable weaker concept to derive the conclusions of Theorem~\ref{th:principal}. We sketch this for the case when the types disagree on the drift but agree on the volatility $\sigma:=\sigma_{i}$.
  
  Let $Q_{0}$ be a probability under which 
  $$
    dX(t)=\sigma(t,X(t))\,dW^{0}(t)
  $$
  where $W^{0}$ is a $Q_{0}$-Brownian motion. For $1\leq i\leq n$, let $Q_{i}$ be an equivalent probability such that $dW^{i}(t):=dW^{0}(t)-\sigma^{-1}(t,X(t))b_{i}(t,X(t))\,dt$ is a Brownian motion under $Q_{i}$ and thus $dX(t)=b_{i}(t,X(t))\,dt + \sigma(t,X(t))\,dW^{i}(t)$ under $Q_{i}$ as desired for type~$i$.
  Consider under $Q_{0}$ the linear backward stochastic differential equation (BSDE)
  $$
  dY(t)= g(t,X(t),Z(t))\,dt + Z(t)\,dW^{0}(t),\quad Y(T)=f(X(T))
  $$  
  where 
  $$
    g(t,x,z)=\kappa_{\alpha}(t,x) - \mu_{\alpha}(t,x) \sigma(t,x)^{-1} z.
  $$
  This equation has a unique square-integrable solution $(Y,Z)$, cf.~\cite[Proposition~2.2]{ElKarouiPengQuenez.97}, and in fact $Y$ is bounded in the present case. We remark that this solution corresponds to~\eqref{eq:principalLinearPDE} in the sense that if~\eqref{eq:principalLinearPDE} has a smooth solution~$w$ then $Y(t)=w(t,X(t))$ and $Z(t)=\sigma(t,X(t))\partial_{x}w(t,X(t))$. Under the measure $Q_{i}$ of agent~$i$ we have
  \begin{equation}\label{eq:BSDE}
    dY(t)= [g(t,X(t),Z(t)) + Z(t)\sigma^{-1}(t,X(t))b_{i}(t,X(t))]\,dt + Z(t)\,dW^{i}(t).
  \end{equation}
  Similarly as in~\eqref{eq:optPortfolioPrincipal}, this implies that if $Y$ is the price process, the optimal portfolio for agent $i$ is
  $$
    \Phi_{i}(t) = \alpha_{i}(t,X(t)) [g(t,X(t),Z(t))+ Z(t)\sigma^{-1}(t,X(t))b_{i}(t,X(t))].
  $$
  Moreover, the definition of $g$ yields that these portfolios satisfy the market clearing condition $\sum_{i} \Phi_{i}(t)=S(t)$; that is, $P:=Y$ is an equilibrium price process. Conversely, any bounded equilibrium price process~$P'$ induces a square-integrable solution of~\eqref{eq:BSDE} and hence $P'=Y$ by uniqueness.
  
  Since the BSDE is Markovian, one can show that the process $Y$ is necessarily of the form $Y(t)=v_{\alpha}(t,X(t))$ for a deterministic function~$v_{\alpha}$. Even if $v_{\alpha}$ is not necessarily smooth, it is still a viscosity solution of the related PDE and that is sufficient to apply the comparison principle as in part~(ii) of the above proof in order to see that the equilibrium price of Theorem~\ref{th:PDE} dominates~$Y(0)$. Alternately, one can apply the comparison principle of BSDEs; cf.~\cite{ElKarouiPengQuenez.97}.
\end{remark} 

%%%%%%%%%%%%%%%%%%%%%%%%%%%%%%%%%%%%%%%
%%%%%%%%%%%%%%%%%%%%%%%%%%%%%%%%%%%%%%%%%%%%%%%%%%%%%%%%%%%%%%%%%
\bibliography{supply}
\bibliographystyle{mynat}
%%%%%%%%%%%%%%%%%%%%%%%%%%%%%%%%%%%%%%%%%%%%%%%%%%%%%%%%%%%%%%%%%

\end{document}